\documentclass{article}

\usepackage{cmbright}
\usepackage{amssymb,amsmath,amsthm}
\usepackage{mathrsfs}
\usepackage[colorlinks,allcolors=blue]{hyperref}
\usepackage{graphicx}

\def\B{\mathscr B}
\def\C{\mathbb C}
\def\d{\mathrm{d}}
\def\D{\mathscr D}
\def\dom{\mathcal D}
\def\G{\mathcal G}
\def\H{\mathcal H}
\def\h{\frak H}

\def\lone{\mathop{\mathrm{L}^1}\nolimits}
\def\N{\mathbb N}
\def\R{\mathbb R}
\def\S{\mathbb S}
\def\U{\mathrm U}
\def\Z{\mathbb Z}

\def\e{\mathop{\mathrm{e}}\nolimits}

\def\Ran{\mathop{\mathrm{Ran}}\nolimits}
\def\Span{\mathop{\mathrm{Span}}\nolimits}
\DeclareMathOperator*{\slim}{s\hspace{0.1pt}-\hspace{0.1pt}lim}
\DeclareMathOperator*{\wlim}{w\hspace{0.1pt}-\hspace{0.1pt}lim}

\newtheorem{Theorem}{Theorem}[section]
\newtheorem{Remark}[Theorem]{Remark}
\newtheorem{Example}[Theorem]{Example}
\newtheorem{Lemma}[Theorem]{Lemma}
\newtheorem{Corollary}[Theorem]{Corollary}

\newtheorem{Assumption}[Theorem]{Assumption}

\begin{document}

%--------------------------------------------------------------------------------------
% Title
%--------------------------------------------------------------------------------------

\title{Stationary scattering theory for unitary operators\\
with an application to quantum walks}

\author{R. Tiedra de Aldecoa\footnote{Partially supported by the Chilean Fondecyt
Grant 1170008.}}

\date{\small}
\maketitle
\vspace{-1cm}

\begin{quote}
\begin{itemize}
\item[] Facultad de Matem\'aticas, Pontificia Universidad Cat\'olica de Chile,\\
Av. Vicu\~na Mackenna 4860, Santiago, Chile\\
E-mail: rtiedra@mat.uc.cl
\end{itemize}
\end{quote}

%--------------------------------------------------------------------------------------

\begin{abstract}
We present a general account on the stationary scattering theory for unitary operators
in a two-Hilbert spaces setting. For unitary operators $U_0,U$ in Hilbert spaces
$\H_0,\H$ and an identification operator $J:\H_0\to\H$, we give the definitions and
collect properties of the stationary wave operators, the strong wave operators, the
scattering operator and the scattering matrix for the triple $(U,U_0,J)$. In
particular, we exhibit conditions under which the stationary wave operators and the
strong wave operators exist and coincide, and we derive representation formulas for
the stationary wave operators and the scattering matrix. As an application, we show
that these representation formulas are satisfied for a class of anisotropic quantum
walks recently introduced in the literature.
\end{abstract}


\smallskip


%--------------------------------------------------------------------------------------
\tableofcontents
%--------------------------------------------------------------------------------------

%--------------------------------------------------------------------------------------
\section{Introduction and main results}\label{section_intro}
\setcounter{equation}{0}
%--------------------------------------------------------------------------------------

In recent years, there has been a surge of research activity on physical and
mathematical systems described by unitary operators. CMV matrices, Discrete quantum
walks, Koopman operators of dynamical systems, and Floquet operators are some examples
of classes of unitary operators having received a great deal of attention.
Accordingly, in order to have at disposal mathematical tools suited for the study of
unitary operators, various authors have undertaken the task to adapt to the unitary
setup spectral and scattering methods available in the self-adjoint setup. Several
approaches, such as CMV representation of unitary operators
\cite{CMV_2003,CMV_2005,CMV_2006,Sim_2007}, Mourre theory for unitary operators
\cite{ABC15_1,ABC15_2,ABCF06,Bou_2013,FRT13,RST_2018}, time-dependent scattering
theory for unitary operators
\cite{Bes_2011,FMSST_2019,Kap_2010,Kap_2012,RST_2019,ST_2019}, or commutator methods
for unitary dynamical systems \cite{RT19,Sim_2018,Tie15,Tie17_1,Tie18}, have already
been developed to some extent. On the other hand, not much as been done in the case of
stationary scattering theory for unitary operators (see \cite{KK70,Mor19,Yaf92} for
partial results). Our purpose in this paper is to fill in this gap by presenting a
first general account on the stationary scattering theory for unitary operators in a
two-Hilbert spaces setting. Additionally, as an illustration, we show that the theory
presented here applies to a class of anisotropic quantum walks recently introduced in
the literature.

Our main source of inspiration is the monograph \cite{Yaf92}, where the stationary
scattering theory of self-ajoint operators is presented in detail. But we also take
advantage of results on smooth operators for unitary operators, on resolvents of
unitary operators and on quantum walks established in \cite{ABCF06}, \cite{KK70} and
\cite{RST_2018,RST_2019}, respectively.

Here is a description of the content of the paper. In Section \ref{sec_res}, we
collect results on the resolvent and on smooth operators for a unitary operator $U_0$
in a separable Hilbert space $\H_0$. In Lemma \ref{lemma_weak}, we present some
properties of locally $U_0$-smooth and weakly locally $U_0$-smooth operators. In Lemma
\ref{lemma_F_0}, we collect results on decomposable operators in the direct integral
representation of the absolutely continuous part of $U_0$. In Lemma \ref{lemma_HS}, we
present some results on limits of resolvents in the Hilbert-Schmidt topology. Finally,
at the end of the section, we introduce a second unitary operator $U$ in a second
separable Hilbert space $\H$ and an identification operator $J:\H_0\to\H$, and recall
useful formulations of the second resolvent equation for the triple $(U,U_0,J)$.

In Section \ref{section_wave}, we define the stationary wave operators
$w_\pm(U,U_0,J)$ and the strong wave operators $W_\pm(U,U_0,J)$ for the triple
$(U,U_0,J)$, and we exhibit conditions under which both types of wave operators exist
and coincide. In Theorem \ref{thm_stat_wave}, we give conditions guaranteeing the
existence of the stationary wave operators, as well as their representation formulas.
In Theorem \ref{thm_strong_wave}, we present conditions for the existence of the
strong wave operators and their identity with the stationary wave operators. Finally,
in Example \ref{ex_trace_class}, we show that the assumptions of these theorems are
satisfied when the perturbation $V=JU_0-UJ$ is trace class.

In Section \ref{section_matrix}, we define the scattering operator $S(U,U_0,J)$ and
the scattering matrix $S(\theta)$ for the triple $(U,U_0,J)$, and we derive
representation formulas for the scattering matrix (the spectral parameter $\theta$
belongs to an appropriate subset of the spectrum of $U_0$). After proving some
preparatory lemmas, we establish in Theorem \ref{thm_S_matrix} the representation
formulas for the scattering matrix, and we explain how they simplify in the
one-Hilbert space case $\H_0=\H$ and $J=1_{\H_0}$.

Finally, in Section \ref{section_walks}, we apply the theory of the previous sections
to the class of anisotropic quantum walks introduced in \cite{RST_2018,RST_2019}.
First, we show in Lemma \ref{lemma_V_trace} that for these quantum walks the
perturbation $V$ decomposes as a product of Hilbert-Schmidt operators. Then we show in
Theorem \ref{thm_walks} that this implies that the strong wave operators and the
stationary wave operators of the quantum walks coincide, and that the representation
formulas for the stationary wave operators and the scattering matrix are satisfied.

As a conclusion, we point out that the representation formulas obtained in this paper
could be useful to establish various new results. For instance, one could use the
representation formulas of the wave operators as a first step toward the proof of
topological Levinson theorems for unitary scattering systems, as it was done for
self-adjoint scattering systems \cite{Ric_2016}. One could also develop a theory of
higher order resolvent estimates for unitary operators (not yet available in the
literature), and then apply it to the resolvent appearing in the representation
formula for the scattering matrix to infer smoothness or mapping properties of the
scattering operator of unitary scattering systems. This would allow in particular to
prove the existence of quantum time delay for explicit unitary scattering systems,
since mapping properties of the scattering operator are needed for it (see
\cite[Thm.~5.3]{ST_2019}). Finally, having at disposal general representation formulas
for the wave operators and scattering operator of unitary scattering systems could be
helpful to establish results on generalised eigenfunctions and eigenvalues of some
classes of unitary systems.

~\\
\noindent
{\bf Acknowledgements.} The author thanks M. Moscolari for his interest in the project
and a useful discussion about Lemma \ref{lemma_limits}(c).

%--------------------------------------------------------------------------------------
\section{Resolvents and smooth operators}\label{sec_res}
\setcounter{equation}{0}
%--------------------------------------------------------------------------------------

In this section, we collect results on resolvents and smooth operators for unitary
operators in a two-Hilbert spaces setting, starting with the case of one unitary
operator in one Hilbert space.

Let $\H_0$ be a separable Hilbert space with norm $\|\;\!\cdot\;\!\|_{\H_0}$ and inner
product $\langle\;\!\cdot,\;\!\cdot\;\!\rangle_{\H_0}$ linear in the first argument.
Let $U_0$ be a unitary operator in $\H_0$ with spectral decomposition
$$
U_0=\int_0^{2\pi}E^{U_0}(\d\theta)\e^{i\theta},
$$
where $E^{U_0}$ is a real spectral measure on the interval $[0,2\pi)$ (which we
identify with the complex unit circle $\S^1$ when necessary). Let
$\H_{\rm ac}(U_0)\subset\H_0$ be the subspace of absolute continuity of $U_0$,
$P_{\rm ac}(U_0)$ the projection onto $\H_{\rm ac}(U_0)$, and
$E^{U_0}_{\rm ac}:=P_{\rm ac}(U_0)E^{U_0}$ the absolutely continuous part of the
spectral measure $E^{U_0}$. Then, the resolvent of $U_0$ is defined as
$$
R_0(z):=\big(1-zU_0^*\big)^{-1},\quad z\in\C\setminus\S^1.
$$
The resolvent of $U_0$ can be written as a geometric series
\begin{equation}\label{eq_series}
R_0(z)=
\begin{cases}
\sum_{n\ge0}\big(zU_0^*\big)^n & \hbox{if $|z|<1$}\\
-\sum_{n\ge1}\big(z^{-1}U_0\big)^n & \hbox{if $|z|>1$,}
\end{cases}
\end{equation}
and it satisfies the identity $R_0({\bar z}^{-1})^*=-zU_0^*R_0(z)$ relating its values
inside and outside of $\S^1$\;\!:
\begin{center}
\includegraphics[width=155pt]{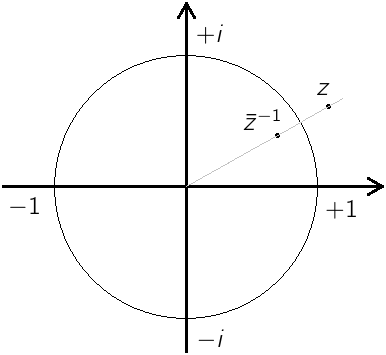}
\end{center}
The resolvent of $U_0$ also satisfies the following relation, which is the analogue of
the first resolvent equation for self-adjoint operators\;\!:
$$
R_0(z_1)-R_0(z_2)=(z_1-z_2)R_0(z_1)U_0^*R_0(z_2),
\quad z_1,z_2\in\C\setminus\S^1.
$$
Setting $z_1=r\e^{i\theta}$ and $z_2=r^{-1}\e^{i\theta}$ for
$r\in(0,\infty)\setminus\{1\}$ and $\theta\in[0,2\pi)$, one gets the relation
$$
R_0(r\e^{i\theta})-R_0(r^{-1}\e^{i\theta})
=\big(1-r^2\big)R_0(r\e^{i\theta})R_0(r\e^{i\theta})^*
=\big(1-r^2\big)\big|R_0(r\e^{i\theta})\big|^2,
$$
which suggests to define the bounded operator \cite[Eq.~(5.2)]{KK70}
\begin{equation}\label{eq_norm}
\delta_0(r,\theta)
:=\tfrac1{2\pi}\big(1-r^2\big)\big|R_0(r\e^{i\theta})\big|^2
\quad\hbox{with norm}\quad
\big\|\delta_0(r,\theta)\big\|_{\B(\H_0)}=\frac{1+r}{2\pi\;\!|1-r|}.
\end{equation}
Here and in the sequel, we write $\B(\H_1)$ for the set of bounded operators on a
Hilbert space $\H_1$, and we write $\B(\H_1,\H_2)$ for the set of bounded operators
from a Hilbert space $\H_1$ to a Hilbert space $\H_2$.

The operator $\delta_0(r,\theta)$ satisfies the symmetry relation
\begin{equation}\label{eq_minus}
\delta_0(r^{-1},\theta)=-\delta_0(r,\theta).
\end{equation}
Using the spectral decomposition of $U_0$ and Fubini's theorem, one gets for
$\varphi_0,\psi_0\in\H_0$
$$
\left\langle\int_0^{2\pi}\d\theta\,\delta_0(r,\theta)\varphi_0,
\psi_0\right\rangle_{\H_0}
=\left\langle\int_0^{2\pi}E^{U_0}(\d\theta')
\left(\tfrac1{2\pi}\int_0^{2\pi}\d\theta\,\big(1-r^2\big)
\big|1-r\e^{i(\theta-\theta')}\big|^{-2}\right)\varphi_0,\psi_0\right\rangle_{\H_0}.
$$
Since the integral of the Poisson kernel in the parenthesis is equal to $1$, namely,
$$
\tfrac1{2\pi}\int_0^{2\pi}\d\theta\,\big(1-r^2\big)
\big|1-r\e^{i(\theta-\theta')}\big|^{-2}=1,
\quad\theta'\in[0,2\pi),
$$
one obtains that
$
\big\langle\int_0^{2\pi}\d\theta\,\delta_0(r,\theta)\varphi_0,\psi_0\big\rangle_{\H_0}
=\langle\varphi_0,\psi_0\rangle_{\H_0}
$,
meaning that we have the (strong integral) identity
\begin{equation}\label{eq_int_delta}
\int_0^{2\pi}\d\theta\,\delta_0(r,\theta)=1,
\quad r\in(0,\infty)\setminus\{1\}.
\end{equation}
We also recall from \cite[Eq.~5.4]{KK70} and \cite[Eq.~1.11.4]{Yaf92} that for any
$\varphi_0,\psi_0\in\H_0$ we have the relations
\begin{equation}\label{eq_derivative}
\lim_{\varepsilon\searrow0}
\big\langle\delta_0\big((1-\varepsilon)^{\pm1},\theta\big)\varphi_0,
\psi_0\big\rangle_{\H_0}
=\pm\tfrac\d{\d\theta}\;\!\big\langle E^{U_0}_{\rm ac}\big((0,\theta]\big)\varphi_0,
\psi_0\big\rangle_{\H_0}
=\pm\tfrac\d{\d\theta}\;\!\big\langle E^{U_0}\big((0,\theta]\big)\varphi_0,
\psi_0\big\rangle_{\H_0}
\end{equation}
for a.e. $\theta\in[0,2\pi)$. But, we emphasise that these relations do not imply the
existence of the weak limits of $\delta_0\big((1-\varepsilon)^{\pm1},\theta\big)$ as
$\varepsilon\searrow0$. In fact, the divergence of the norm in \eqref{eq_norm} as
$r\to1$ shows that these limits do not exist. Yet, \eqref{eq_derivative} implies that
for each $\varphi_0\in\H_0$ there exists a function
$c_{\varphi_0}:[0,2\pi)\to[0,\infty)$ such that
\begin{equation}\label{eq_bound_res}
\pm\tfrac1{2\pi}\big(1-(1-\varepsilon)^{\pm2}\big)
\big\|R_0\big((1-\varepsilon)^{\pm1}\e^{i\theta}\big)\varphi_0\big\|_{\H_0}^2
=\big\langle\delta_0(1-\varepsilon,\theta)\varphi_0,\varphi_0\big\rangle_{\H_0}
\le c_{\varphi_0}(\theta)
\end{equation}
for all $\varepsilon\in(0,1)$ and a.e. $\theta\in[0,2\pi)$.

If $\G$ is an auxiliary separable Hilbert space, then an operator $T_0\in\B(\H_0,\G)$
is called locally $U_0$-smooth on a Borel set $\Theta\subset[0,2\pi)$ if there exists
$c_\Theta\ge0$ such that
\begin{equation}\label{def_U_smooth}
\sum_{n\in\Z}\big\|T_0\;\!U_0^nE^{U_0}(\Theta)\varphi\big\|_\G^2
\le c_\Theta\;\!\|\varphi\|_{\H_0}^2\quad\hbox{for all $\varphi\in\H_0$},
\end{equation}
and $T_0$ is called $U_0$-smooth if \eqref{def_U_smooth} is satisfied with
$\Theta=[0,2\pi)$. Similarly, $T_0$ is called weakly locally $U_0$-smooth on a Borel
set $\Theta\subset[0,2\pi)$ if the weak limit
\begin{equation}\label{def_weak}
\wlim_{\varepsilon\searrow0}T_0\;\!\delta_0(1-\varepsilon,\theta)E^{U_0}(\Theta)T_0^*
\hbox{~exists for a.e. $\theta\in[0,2\pi)$,}
\end{equation}
and $T_0$ is called weakly $U_0$-smooth if \eqref{def_weak} is satisfied with
$\Theta=[0,2\pi)$. Some properties of locally $U_0$-smooth and weakly locally
$U_0$-smooth operators are collected in the following lemma. We use the notation
$\G^*$ for the adjoint space of $\G$, and we set
$$
g_\pm(\varepsilon):=\tfrac1{2\pi}\big(1-(1-\varepsilon)^{\pm2}\big),
\quad\varepsilon\in(0,1).
$$

\begin{Lemma}[$U_0$-smooth operators]\label{lemma_weak}
Let $T_0\in\B(\H_0,\G)$ and let $\Theta\subset[0,2\pi)$ be a Borel set.
\begin{enumerate}
\item[(a)] $T_0$ is weakly locally $U_0$-smooth on $\Theta$ if and only if any of the
following estimates is satisfied for some function $c:[0,2\pi)\to[0,\infty)$ and all
$\varepsilon\in(0,1)$
\begin{align}
\big\|T_0\;\!\delta_0(1-\varepsilon,\theta)E^{U_0}(\Theta)T_0^*\big\|_{\B(\G)}
\le c(\theta),\quad\hbox{a.e. $\theta\in[0,2\pi)$,}\label{eq_estimate_1}\\
\big|g_\pm(\varepsilon)\big|^{1/2}\;\!
\big\|T_0R_0\big((1-\varepsilon)^{\pm1}\e^{i\theta}\big)
E^{U_0}(\Theta)\big\|_{\B(\H_0,\G)}
\le c(\theta)^{1/2},\quad\hbox{a.e. $\theta\in[0,2\pi)$,}
\label{eq_estimate_2}\\
\big|g_\pm(\varepsilon)\big|^{1/2}\;\!
\big\|T_0R_0\big((1-\varepsilon)^{\pm1}\e^{i\theta}\big)^*
E^{U_0}(\Theta)\big\|_{\B(\H_0,\G)}
\le c(\theta)^{1/2},\quad\hbox{a.e. $\theta\in[0,2\pi)$.}
\label{eq_estimate_3}
\end{align}
\item[(b)] If $T_0$ is weakly locally $U_0$-smooth on $\Theta$ and $\varphi_0\in\H_0$,
then the weak limit
$
\wlim_{\varepsilon\searrow0}
T_0\;\!\delta_0(1-\varepsilon,\theta)E^{U_0}(\Theta)\varphi_0
$
exists for a.e. $\theta\in[0,2\pi)$.
\item[(c)] Assume that for each $\varphi_0$ in a dense set $\D_0\subset\H_0$ the weak
limit $\wlim_{\varepsilon\searrow0}T_0\;\!\delta_0(1-\varepsilon,\theta)\varphi_0$
exists for a.e. $\theta\in[0,2\pi)$. Then, there exists a set
$\D_0'\subset\H_{\rm ac}(U_0)$ dense in $\H_{\rm ac}(U_0)$ such that
$$
\sum_{n\in\Z}\big\|T_0U_0^n\psi_0\big\|_\G^2<\infty
\hbox{~for all $\psi_0\in\D_0'$.}
$$
\item[(d)] If $T_0$ is locally $U_0$-smooth on $\Theta$, then $T_0$ is weakly locally
$U_0$-smooth on $\Theta$.
\item[(e)] If $T_0$ is locally $U_0$-smooth on $\Theta$, then
$\overline{E^{U_0}(\Theta)T_0^*\;\!\G^*}\subset\H_{\rm ac}(U_0)$.
\end{enumerate}
\end{Lemma}

\begin{proof}
(a) Using the definition \eqref{eq_norm} of $\delta_0(1-\varepsilon,\theta)$ and the
identity $\|BB^*\|_{\B(\G)}=\|B\|^2_{\B(\H_0,\G)}$ for operators $B\in\B(\H_0,\G)$, we
obtain the equivalence of the estimates \eqref{eq_estimate_1}-\eqref{eq_estimate_3}.
Furthermore, the estimate \eqref{eq_estimate_1} follows directly from
\eqref{def_weak}. So, to prove the claim it only remains to show that
\eqref{eq_estimate_1} implies \eqref{def_weak}, which can be done as in the proof of
\cite[Lemma~5.1.2]{Yaf92}.

(b) One can show as in \cite[Lemma~5.1.6]{Yaf92} that it is sufficient to prove that
$\|T_0\;\!\delta_0(1-\varepsilon,\theta)E^{U_0}(\Theta)\varphi_0\|_\G$ is bounded
uniformly in $\varepsilon$ for a.e. $\theta\in[0,2\pi)$ to establish the existence of
the limit. Now, since
\begin{align*}
&\big\|T_0\;\!\delta_0(1-\varepsilon,\theta)E^{U_0}(\Theta)\varphi_0\big\|_\G\\
&\le\left(g_+(\varepsilon)^{1/2}\;\!
\big\|T_0R_0\big((1-\varepsilon)\e^{i\theta}\big)^*
E^{U_0}(\Theta)\big\|_{\B(\H_0,\G)}\right)
\left(g_+(\varepsilon)^{1/2}\;\!
\big\|R_0\big((1-\varepsilon)\e^{i\theta}\big)\varphi_0\big\|_{\H_0}\right)
\end{align*}
for all $\varepsilon\in(0,1)$ and $\theta\in[0,2\pi)$, this follows from
\eqref{eq_bound_res} and \eqref{eq_estimate_3}.

(c) For each $\varphi_0\in\D_0$, set
$
F_{\varphi_0}(\theta)
:=\wlim_{\varepsilon\searrow0}T_0\;\!\delta_0(1-\varepsilon,\theta)\varphi_0
$
for a.e. $\theta\in[0,2\pi)$ and define the sets
$$
\Theta_{\varphi_0,N}
:=\big\{\theta\in[0,2\pi)\mid\|F_{\varphi_0}(\theta)\|_\G\le N\big\},
\quad N\in\N.
$$
Since we have $\lim_{N\to\infty}|[0,2\pi)\setminus\Theta_{\varphi_0,N}|=0$,  the set
$$
\D_0':=\Span\big\{E^{U_0}_{\rm ac}(\Theta_{\varphi_0,N})\varphi_0
\mid\varphi_0\in\D_0,~N\in\N\big\}\subset\H_{\rm ac}(U_0)
$$
is dense in $\H_{\rm ac}(U_0)$. Let $\{\frak g_k\}_{k\in\N}$ be an orthonormal basis
of $\G$ and take $\psi_0:=E^{U_0}_{\rm ac}(\Theta_{\varphi_0,N})\varphi_0\in\D_0'$.
Then, using the spectral theorem and \eqref{eq_derivative}, we get for any $k\in\N$
\begin{align*}
\big\langle T_0U_0^n\psi_0,\frak g_k\big\rangle_\G
&=\int_{\Theta_{\varphi_0,N}}\d\theta\,\e^{in\theta}\tfrac\d{\d\theta}\;\!
\big\langle E^{U_0}((0,\theta])\varphi_0,T_0^*\frak g_k\big\rangle_{\H_0}\\
&=\int_{\Theta_{\varphi_0,N}}\d\theta\,\e^{in\theta}\lim_{\varepsilon\searrow0}
\big\langle T_0\;\!\delta_0(1-\varepsilon,\theta)\varphi_0,\frak g_k\big\rangle_\G\\
&=\int_{\Theta_{\varphi_0,N}}\d\theta\,\e^{in\theta}
\big\langle F_{\varphi_0}(\theta),\frak g_k\big\rangle_\G.
\end{align*}
Applying Parseval's identity for the orthonormal basis
$\{\frak g_k\}_{k\in\N}$, we then obtain that
$$
\sum_{n\in\Z}\big\|T_0U_0^n\psi_0\big\|_\G^2
=\sum_{n\in\Z}\sum_{k\in\N}\left|\int_{\Theta_{\varphi_0,N}}\d\theta\,\e^{in\theta}
\big\langle F_{\varphi_0}(\theta),\frak g_k\big\rangle_\G\right|^2
=2\pi\int_{\Theta_{\varphi_0,N}}\d\theta\,\big\|F_{\varphi_0}(\theta)\big\|_\G^2
\le(2\pi N)^2,
$$
which proves the claim for vectors of the form
$\psi_0=E^{U_0}_{\rm ac}(\Theta_{\varphi_0,N})\varphi_0$.
The claim for arbitrary vectors $\psi_0\in\D_0'$ can now be deduced using this
bound, Cauchy-Schwartz inequality, and the identity
$$
\left(\sum_{i=1}^ma_i\right)^2=\sum_{i=1}^ma_i^2+2\sum_{i<j}a_ia_j,
\quad m\in\N^*,~a_i\in\R.
$$

(d) One can show as in \cite[Appendix~A]{ABCF06} that \eqref{def_U_smooth} is
equivalent to
$$
\sup_{\varepsilon\in(0,1],\,\theta\in[0,2\pi),\,\zeta\in\G,\,\|\varphi\|_\G=1}
\big|\big\langle T_0\;\!\delta_0(1-\varepsilon,\theta)E^{U_0}(\Theta)T_0^*\zeta,
\zeta\big\rangle_\G\big|<\infty.
$$
Since $T_0\;\!\delta_0(1-\varepsilon,\theta)E^{U_0}(\Theta)T_0^*$ is bounded and
self-adjoint in $\G$ for all $\varepsilon\in(0,1]$ and $\theta\in[0,2\pi)$, this in
turn is equivalent to
$$
\sup_{\varepsilon\in(0,1],\,\theta\in[0,2\pi)}
\big\|T_0\;\!\delta_0(1-\varepsilon,\theta)E^{U_0}(\Theta)T_0^*\big\|_{\B(\G)}<\infty.
$$
Therefore, the estimate \eqref{eq_estimate_1} is satisfied for all $\theta\in[0,2\pi)$
with a function $c$ independent of $\theta$. It thus follows by point (a) that $T_0$
is weakly locally $U_0$-smooth on $\Theta$.

(e) This can be proved as in \cite[Thm.~2.1]{ABCF06}.
\end{proof}

Let $\widehat\sigma_0$ be a core of the spectrum of $U_0$, that is, a Borel set
$\widehat\sigma_0\subset[0,2\pi)$ such that (i) $\widehat\sigma_0$ is a Borel support
of $E^{U_0}$, i.e. $E^{U_0}\big([0,2\pi)\setminus\widehat\sigma_0\big)=0$, and (ii) if
$\Theta$ is a Borel support of $E^{U_0}$, then $\widehat\sigma_0\setminus\Theta$ has
Lebesgue measure $0$. Then, there exist for a.e. $\theta\in\widehat\sigma_0$ separable
Hilbert spaces $\h_0(\theta)$ and an operator (a spectral transformation)
\begin{equation}\label{def_F_0}
F_0:\H_0\to\int_{\widehat\sigma_0}^\oplus\d\theta\,\h_0(\theta),
\end{equation}
which is unitary from $\H_{\rm ac}(U_0)$ to
$\int_{\widehat\sigma_0}^\oplus\d\theta\,\h_0(\theta)$, vanishes on the singular
continuous subspace $\H_{\rm sc}(U_0)$ of $U_0$, and diagonalises the restriction
$U_0\upharpoonright\H_{\rm ac}(U_0)$. Namely, if we set
$$
U_0^{\rm(ac)}:=U_0\upharpoonright\H_{\rm ac}(U_0)
\quad\hbox{and}\quad
F_0^{\rm(ac)}:=F_0\upharpoonright\H_{\rm ac}(U_0),
$$
then we have the direct integral decomposition
$$
F_0^{\rm(ac)}U_0^{\rm(ac)}\big(F_0^{\rm(ac)}\big)^*
=\int_{\widehat\sigma_0}^\oplus\d\theta\,\e^{i\theta}.
$$
In particular, we get for any Borel set $\Theta\subset[0,2\pi)$ and
$\varphi_0,\psi_0\in\H_0$ that
$$
\big\langle E^{U_0}_{\rm ac}(\Theta)\varphi_0,\psi_0\big\rangle_{\H_0}
=\int_{\widehat\sigma_0\cap\Theta}\d\theta\,\big\langle(F_0\varphi_0)(\theta),
(F_0\psi_0)(\theta)\big\rangle_{\h_0(\theta)}.
$$
Since $\Theta$ is arbitrary, we infer from this and \eqref{eq_derivative} that
\begin{equation}\label{eq_derivative_bis}
\lim_{\varepsilon\searrow0}\big\langle\delta_0\big((1-\varepsilon)^{\pm1},
\theta\big)\varphi_0,\psi_0\big\rangle_{\H_0}
=\pm\big\langle(F_0\varphi_0)(\theta),(F_0\psi_0)(\theta)\big\rangle_{\h_0(\theta)},
\quad\hbox{a.e. $\theta\in\widehat\sigma_0$.}
\end{equation}

Some properties of decomposable operators in the direct integral
$\int_{\widehat\sigma_0}^\oplus\d\theta\,\h_0(\theta)$ are collected in the following
lemma.

\begin{Lemma}\label{lemma_F_0}
\begin{enumerate}
\item[(a)] Let $a(\theta)\in\B\big(\h_0(\theta)\big)$ for a.e.
$\theta\in\widehat\sigma_0$ and let $\D_0\subset\H_0$ be a dense set. Assume that for
any $\varphi_0,\psi_0\in\D_0$ there exists a Borel set ${\cal B}\subset\widehat\sigma_0$ of
full measure in $\widehat\sigma_0$ (depending on $\varphi_0,\psi_0$) such that
$$
\big\langle a(\theta)(F_0\varphi_0)(\theta),
(F_0\psi_0)(\theta)\big\rangle_{\h_0(\theta)}=0
\quad\hbox{for all $\theta\in{\cal B}$.}
$$
Then, $a(\theta)=0$ for a.e. $\theta\in\widehat\sigma_0$.

\item[(b)] Let $T_0\in\B(\H_0,\G)$ be weakly $U_0$-smooth, let
$\Delta(\theta,T_0)\in\B(\G)$ be the weak operator limit obtained in \eqref{def_weak}
for a.e. $\theta\in[0,2\pi)$ when $\Theta=[0,2\pi)$, and let
$$
Z_0(\theta,T_0)\zeta:=\big(F_0T_0^*\zeta\big)(\theta),
\quad\hbox{$\zeta\in\G$, a.e. $\theta\in\widehat\sigma_0$.}
$$
Then, for each $\theta$ in a Borel set ${\cal B}\subset\widehat\sigma_0$ of full measure in
$\widehat\sigma_0$, we have $Z_0(\theta,T_0)\in\B\big(\G,\h_0(\theta)\big)$ and
$$
Z_0(\theta,T_0)^*\;\!Z_0(\theta,T_0)=\Delta(\theta,T_0).
$$
\item[(c)] Let $A\in\B(\G)$ and assume that $T_0,T_1\in\B(\H_0,\G)$ are weakly
$U_0$-smooth. Then, the operator $P_{\rm ac}(U_0)T_0^*A T_1P_{\rm ac}(U_0)$ is an
integral operator in $\int_{\widehat\sigma_0}^\oplus\d\theta\,\h_0(\theta)$, in the
sense that
\begin{align*}
&\big\langle P_{\rm ac}(U_0)T_0^*A T_1P_{\rm ac}(U_0)\varphi_0,
\psi_0\big\rangle_{\H_0}\\
&=\int_{\widehat\sigma_0}\d\mu\int_{\widehat\sigma_0}\d\nu\,
\big\langle Z_0(\mu,T_0)A\;\!Z_0(\nu,T_1)^*(F_0\varphi_0)(\nu),
(F_0\psi_0)(\mu)\big\rangle_{\h_0(\mu)},\quad\varphi_0,\psi_0\in\H_0.
\end{align*}
Furthermore, if (any) one the following strong limits exists
\begin{equation}\label{eq_strong_lim}
\slim_{\varepsilon\searrow0}T_1\delta_0(1-\varepsilon,\theta)\varphi_0,
\quad\slim_{\varepsilon\searrow0}T_0\delta_0(1-\varepsilon,\theta)\psi_0,
\quad\hbox{$\varphi_0,\psi_0\in\H_0$, a.e. $\theta\in\widehat\sigma_0$,}
\end{equation}
then for a.e. $\mu\in\widehat\sigma_0$ and a.e. $\nu\in\widehat\sigma_0$ there exists
the double limit
\begin{align}
&\lim_{\varepsilon\searrow0,\,\rho\searrow0}\big\langle
P_{\rm ac}(U_0)T_0^*A T_1P_{\rm ac}(U_0)\delta_0(1-\varepsilon,\mu)\varphi_0,
\delta_0(1-\rho,\nu)\psi_0\big\rangle_{\H_0}\nonumber\\
&=\big\langle Z_0(\mu,T_0)A\;\!Z_0(\nu,T_1)^*(F_0\varphi_0)(\nu),
(F_0\psi_0)(\mu)\big\rangle_{\h_0(\mu)}.\label{eq_lim_kernel}
\end{align}
Finally, if both limits in \eqref{eq_strong_lim} exist and
$\{A(\tau)\}_{\tau\in(0,1)}\subset\B(\G)$ satisfies $\wlim_{\tau\searrow0}A(\tau)=A$,
then for a.e. $\mu\in\widehat\sigma_0$ and a.e. $\nu\in\widehat\sigma_0$ there exists
the triple limit
\begin{align}
&\lim_{\varepsilon\searrow0,\,\rho\searrow0,\,\tau\searrow0}\big\langle
P_{\rm ac}(U_0)T_0^*A(\tau)T_1P_{\rm ac}(U_0)\delta_0(1-\varepsilon,\mu)\varphi_0,
\delta_0(1-\rho,\nu)\psi_0\big\rangle_{\H_0}\nonumber\\
&=\big\langle Z_0(\mu,T_0)A\;\!Z_0(\nu,T_1)^*(F_0\varphi_0)(\nu),
(F_0\psi_0)(\mu)\big\rangle_{\h_0(\mu)}.\label{eq_lim_kernel_bis}
\end{align}
\end{enumerate}
\end{Lemma}

\begin{proof}
(a) The claim can be shown as in \cite[Lemma~1.5.1]{Yaf92}.

(b) Let ${\cal B}\subset\widehat\sigma_0$ be the Borel set of full measure in
$\widehat\sigma_0$ for which the limit \eqref{def_weak} with $\Theta=[0,2\pi)$ exists.
Then, using \eqref{eq_derivative_bis}, we obtain for any $\zeta,\xi\in\G$ and
$\theta\in{\cal B}$
\begin{align*}
\big\langle Z_0(\theta,T_0)^*\;\!Z_0(\theta,T_0)\zeta,\xi\big\rangle_{\h_0(\theta)}
&=\big\langle\big(F_0T_0^*\zeta\big)(\theta),\big(F_0T_0^*\xi\big)(\theta)
\big\rangle_{\h_0(\theta)}\\
&=\lim_{\varepsilon\searrow0}\big\langle T_0\;\!
\delta_0\big((1-\varepsilon),\theta\big)T_0^*\zeta,\xi\big\rangle_\G\\
&=\big\langle\Delta(\theta,T_0)\zeta,\xi\big\rangle_\G.
\end{align*}
Since $\zeta$ and $\xi$ are arbitrary and $\Delta(\theta,T_0)$ is bounded, it follows
that $Z_0(\theta,T_0)\in\B\big(\G,\h_0(\theta)\big)$ and
$$
Z_0(\theta,T_0)^*\;\!Z_0(\theta,T_0)=\Delta(\theta,T_0).
$$

(c) The claim can be shown as in the self-adjoint case \cite[Lemma~5.4.7]{Yaf92} (some
steps are shorter here in the unitary case since the operators $T_0,T_1$ are bounded,
which is not the case in \cite{Yaf92}).
\end{proof}

In the next lemma, we present some results on limits of resolvents in the
Hilbert-Schmidt topology. We use the notation $\sigma_{\rm p}(U_0)$ for the point
spectrum of $U_0$ and $S_2(\H_0,\G)$ (resp. $S_2(\G)$) for the set of Hilbert-Schmidt
operators from $\H_0$ to $\G$ (resp. from $\G$ to $\G$).

\begin{Lemma}[Hilbert-Schmidt case]\label{lemma_HS}
Let $T_0,T_1\in S_2(\H_0,\G)$, $\varepsilon\in(0,1)$, and $\theta\in[0,2\pi)$.
\begin{enumerate}
\item[(a)] Let $H_0$ be a self-adjoint operator in $\H_0$. Then, the operators
$$
T_0\left(H_0-i\;\!\tfrac{1+(1-\varepsilon)^{\pm1}\e^{i\theta}}
{1-(1-\varepsilon)^{\pm1}\e^{i\theta}}\right)^{-1}T_0^*\in S_2(\G),
$$
have a limit in $S_2(\G)$ as $\varepsilon\searrow0$ for a.e. $\theta\in[0,2\pi)$.

\item[(b)] The operators
$
T_0R_0\big((1-\varepsilon)^{\pm1}\e^{i\theta}\big)T_1^*\in S_2(\G)
$
have a limit in $S_2(\G)$ as $\varepsilon\searrow0$ for a.e. $\theta\in[0,2\pi)$.

\item[(c)] For any $\varphi_0\in\H_0$, the vectors
$
T_0R_0\big((1-\varepsilon)^{\pm1}\e^{i\theta}\big)\varphi_0\in\G
$
have a strong limit in $\G$ as $\varepsilon\searrow0$ for a.e. $\theta\in[0,2\pi)$.
\end{enumerate}
\end{Lemma}

\begin{proof}
(a) The proof is the similar to that of \cite[Thm.~6.1.9]{Yaf92} with the following
two modifications\;\!: (i) The radial limits as $\varepsilon\searrow0$ of functions of
$\lambda\pm i\varepsilon$ have to be replaced by angular limits as
$\varepsilon\searrow0$ of functions of
$
i\;\!\tfrac{1+(1-\varepsilon)^{\pm1}\e^{i\theta}}
{1-(1-\varepsilon)^{\pm1}\e^{i\theta}}
$.
This is possible, because the theorems used in the proof of \cite[Thm.~6.1.9]{Yaf92}
also hold for angular limits (see theorems 1.2.2, 1.2.3 and 1.2.5 of \cite{Yaf92}).
(ii) The equation at the bottom of page 191 of \cite{Yaf92}
$$
iB^*-iB=2\varepsilon\;\!G_1RR^*G_1^*\ge0
$$
has to be replaced by the equation
$$
iB^*-iB
=2\;\!\tfrac{1-|(1-\varepsilon)\e^{i\theta}|^2}
{|1-(1-\varepsilon)\e^{i\theta}|^2}G_1RR^*G_1^*
\ge0.
$$

(b) Since $\H_0$ is separable, there exists $\phi\in[0,2\pi)$ such that
$1\notin\sigma_{\rm p}(\e^{i\phi}U_0)$. It follows that the Cayley transform of
$\e^{i\phi}U_0$
$$
H_0:=i\;\!\big(1+\e^{i\phi}U_0\big)\big(1-\e^{i\phi}U_0\big)^{-1},
\quad\dom(H_0):=\Ran\big(1-\e^{i\phi}U_0\big),
$$
is a self-adjoint operator satisfying the relation
$$
R_0(z)=\big(1-\e^{i\phi}z\big)^{-1}+\tfrac{2i\e^{i\phi}z}{(1-\e^{i\phi}z)^2}
\left(H_0-i\;\!\tfrac{1+\e^{i\phi}z}{1-\e^{i\phi}z}\right)^{-1},
\quad z\in\C\setminus\S^1.
$$
In particular, in the case $z=(1-\varepsilon)^{\pm1}\e^{i\theta}$ we get
\begin{align*}
&T_0R_0\big((1-\varepsilon)^{\pm1}\e^{i\theta}\big)T_0^*\\
&=\big(1-(1-\varepsilon)^{\pm1}\e^{i(\phi+\theta)}\big)^{-1}T_0T_0^*
+\tfrac{2i(1-\varepsilon)^{\pm1}\e^{i(\phi+\theta)}}
{(1-(1-\varepsilon)^{\pm1}\e^{i(\phi+\theta)})^2}
T_0\left(H_0-i\;\!\tfrac{1+(1-\varepsilon)^{\pm1}\e^{i(\phi+\theta)}}
{1-(1-\varepsilon)^{\pm1}\e^{i(\phi+\theta)}}\right)^{-1}T_0^*.
\end{align*}
Therefore, it follows from point (a) that the operators
$
T_0R_0\big((1-\varepsilon)^{\pm1}\e^{i\theta}\big)T_0^*\in S_2(\G)
$
have a limit in $S_2(\G)$ as $\varepsilon\searrow0$ for a.e. $\theta\in[0,2\pi)$. The
claim then follows from this by using the polarization identity\;\!:
\begin{align*}
T_0R_0(z)T_1^*
&=(T_0+T_1)R_0(z)(T_0+T_1)^*-(T_0-T_1)R_0(z)(T_0-T_1)^*\\
&\quad-i(T_0-iT_1)R_0(z)(T_0-iT_1)^*+i(T_0+iT_1)R_0(z)(T_0+iT_1)^*,
\quad z\in\C\setminus\S^1.
\end{align*}

(c) Let $\zeta\in\G$ and take an operator $T_1\in S_2(\H_0,\G)$ satisfying
$T_1^*\zeta=\varphi_0$. Then, we have
$$
T_0R_0\big((1-\varepsilon)^{\pm1}\e^{i\theta}\big)\varphi_0
=T_0R_0\big((1-\varepsilon)^{\pm1}\e^{i\theta}\big)T_1^*\zeta,
$$
and the claim follows from point (b).
\end{proof}

Now, let $U$ be a second unitary operator in a second separable Hilbert space $\H$,
and define the quantities $E^U$, $R(z)$, $\delta(r,\theta)$, etc. as for the operator
$U_0$. Then, the above results for the operator $U_0$ also hold for the operator $U$.
Furthermore, if we are given some identification operator $J\in\B(\H_0,\H)$, then a
direct calculation shows that
$$
JR_0(z)-R(z)J=-zR(z)U^*(JU_0-UJ)U_0^*R_0(z),\quad z\in\C\setminus\S^1.
$$
So, if one sets $V:=JU_0-UJ$, then one obtains a second resolvent equation for unitary
operators in a two-Hilbert spaces setting\;\!:
\begin{equation}\label{eq_factors}
JR_0(z)-R(z)J=-zR(z)U^*VU_0^*R_0(z),\quad z\in\C\setminus\S^1.
\end{equation}
By analogy to the self-adjoint case, we call the operator $V$ the perturbation of
$U_0$. And if $V$ factorises as $V=G^*G_0$ for some operators $G_0\in\B(\H_0,\G)$ and
$G\in\B(\H,\G)$, then the second resolvent equation takes the form
\begin{equation}\label{eq_factors_bis}
JR_0(z)-R(z)J=-z\big(GUR(z)^*\big)^*G_0U_0^*R_0(z),\quad z\in\C\setminus\S^1.
\end{equation}

%--------------------------------------------------------------------------------------
\section{Representation formulas for the wave operators}\label{section_wave}
\setcounter{equation}{0}
%--------------------------------------------------------------------------------------

In this section, we define the stationary wave operators and the strong wave operators
for the triple $(U,U_0,J)$, we present conditions under which both types of wave
operators exist and coincide, and we derive representation formulas for the stationary
wave operators. We start with the definition of the stationary wave operators.

For each $\varepsilon\in(0,1)$, we define operators
$w_\pm(U,U_0,J,\varepsilon)\in\B(\H_0,\H)$ in terms of the sesquilinear form
\begin{align*}
&\big\langle w_\pm(U,U_0,J,\varepsilon)\varphi_0,\varphi\big\rangle_\H\\
&:=\pm g_\pm(\varepsilon)\int_0^{2\pi}\d\theta\,
\big\langle JR_0\big((1-\varepsilon)^{\pm1}\e^{i\theta}\big)\varphi_0,
R\big((1-\varepsilon)^{\pm1}\e^{i\theta}\big)\varphi\big\rangle_\H,
\quad\varphi_0\in\H_0,~\varphi\in\H.
\end{align*}
Using the Cauchy-Schwarz inequality and \eqref{eq_int_delta}, we obtain the bound
\begin{align}
&\Big|\big\langle w_\pm(U,U_0,J,\varepsilon)\varphi_0,\varphi\big\rangle_\H\Big|^2
\nonumber\\
&\le\|J\|^2_{\B(\H_0,\H)}\left(g_\pm(\varepsilon)\int_0^{2\pi}\d\theta\,
\big\|R_0\big((1-\varepsilon)^{\pm1}\e^{i\theta}\big)
\varphi_0\big\|^2_{\B(\H_0)}\right)\nonumber\\
&\quad\cdot\left(g_\pm(\varepsilon)\int_0^{2\pi}\d\theta\,
\big\|R\big((1-\varepsilon)^{\pm1}\e^{i\theta}\big)\varphi\big\|^2_{\B(\H)}\right)
\nonumber\\
&=\|J\|^2_{\B(\H_0,\H)}\left(\int_0^{2\pi}\d\theta\,
\big\langle\delta_0\big((1-\varepsilon)^{\pm1},\theta\big)\varphi_0,
\varphi_0\big\rangle_{\H_0}\right)\left(\int_0^{2\pi}\d\theta\,
\big\langle\delta\big((1-\varepsilon)^{\pm1},\theta\big)\varphi,
\varphi\big\rangle_\H\right)\nonumber\\
&=\|J\|^2_{\B(\H_0,\H)}\;\!\|\varphi_0\|^2_{\H_0}\;\!\|\varphi\|^2_\H,
\label{eq_bound_w}
\end{align}
which shows that the operators $w_\pm(U,U_0,J,\varepsilon)$ are bounded uniformly in
$\varepsilon$\;\!:
\begin{equation}\label{eq_bound_uniform}
\big\|w_\pm(U,U_0,J,\varepsilon)\big\|_{\B(\H_0,\H)}\le\|J\|_{\B(\H_0,\H)},
\quad\varepsilon\in(0,1).
\end{equation}

In our first lemma of the section, we present a general condition for the existence of
the weak limits of $w_\pm(U,U_0,J,\varepsilon)$ as $\varepsilon\searrow0$. We use the
notation $\chi_{\cal B}$ for the characteristic function of a Borel set
${\cal B}\subset[0,2\pi)$.

\begin{Lemma}\label{lemma_limits}
Let $\Theta_0,\Theta\subset[0,2\pi)$ be Borel sets, let $\varphi_0\in\H_0$ and
$\varphi\in\H$, and assume that the limits
\begin{equation}\label{del_a_pm}
a_\pm(\varphi_0,\varphi,\theta)
:=\pm\lim_{\varepsilon\searrow0}g_\pm(\varepsilon)\;\!
\big\langle JR_0\big((1-\varepsilon)^{\pm1}\e^{i\theta}\big)\varphi_0,
R\big((1-\varepsilon)^{\pm1}\e^{i\theta}\big)\varphi\big\rangle_\H
\end{equation}
exist for a.e. $\theta\in\Theta_0\cap\Theta$.
\begin{enumerate}
\item[(a)] We have
\begin{equation}\label{eq_null}
a_\pm\big(E^{U_0}(\Theta_0)\varphi_0,E^U(\Theta)\varphi,\theta\big)=0,
\quad\hbox{a.e.}~\theta\in[0,2\pi)\setminus(\Theta_0\cap\Theta),
\end{equation}
and
\begin{equation}\label{eq_intersection}
a_\pm\big(E^{U_0}(\Theta_0)\varphi_0,E^U(\Theta)\varphi,\theta\big)
=\chi_{\Theta_0\cap\Theta}(\theta)\;\!a_\pm(\varphi_0,\varphi,\theta),
\quad\hbox{a.e.}~\theta\in[0,2\pi).
\end{equation}
\item[(b)] If $a_\pm(\varphi_0,\varphi,\theta)$ exist for a.e. $\theta\in[0,2\pi)$,
then the same holds with $\varphi_0$ replaced by $P_{\rm ac}(U_0)\varphi_0$ or with
$\varphi$ replaced by $P_{\rm ac}(U)\varphi$, and we have for a.e. $\theta\in[0,2\pi)$
\begin{equation}\label{eq_ac}
a_\pm(\varphi_0,\varphi,\theta)
=a_\pm\big(P_{\rm ac}(U_0)\varphi_0,\varphi,\theta\big)
=a_\pm\big(\varphi_0,P_{\rm ac}(U)\varphi,\theta\big)
=a_\pm\big(P_{\rm ac}(U_0)\varphi_0,P_{\rm ac}(U)\varphi,\theta\big).
\end{equation}
\item[(c)] We have
$$
\lim_{\varepsilon\searrow0}\big\langle w_\pm(U,U_0,J,\varepsilon)E^{U_0}(\Theta_0)
P_{\rm ac}(U_0)\varphi_0,E^U(\Theta)\varphi\big\rangle_\H
=\int_{\Theta_0\cap\Theta}\d\theta\,a_\pm(\varphi_0,\varphi,\theta).
$$
\end{enumerate}
\end{Lemma}

\begin{proof}
(a) Using successively the Cauchy-Schwarz inequality, the definition of $\delta_0$ and
$\delta$, the relation \eqref{eq_derivative}, and functional calculus, we obtain for
a.e. $\theta\in[0,2\pi)$
\begin{align*}
&\lim_{\varepsilon\searrow0}\Big|g_\pm(\varepsilon)\;\!
\big\langle JR_0\big((1-\varepsilon)^{\pm1}\e^{i\theta}\big)E^{U_0}(\Theta_0)
\varphi_0,R\big((1-\varepsilon)^{\pm1}\e^{i\theta}\big)E^U(\Theta)\varphi
\big\rangle_\H\Big|^2\\
&\le\lim_{\varepsilon\searrow0}\|J\|^2_{\B(\H_0,\H)}\left(g_\pm(\varepsilon)
\big\|R_0\big((1-\varepsilon)^{\pm1}\e^{i\theta}\big)E^{U_0}(\Theta_0)\varphi_0
\big\|^2_{\B(\H_0)}\right)\\
&\quad\cdot\left(g_\pm(\varepsilon)\big\|R\big((1-\varepsilon)^{\pm1}
\e^{i\theta}\big)E^U(\Theta)\varphi\big\|^2_{\B(\H)}\right)\\
&=\|J\|^2_{\B(\H_0,\H)}\lim_{\varepsilon\searrow0}
\big\langle\delta_0\big((1-\varepsilon)^{\pm1},\theta\big)E^{U_0}(\Theta_0)\varphi_0,
\varphi_0\big\rangle_{\H_0}\;\!\big\langle
\delta\big((1-\varepsilon)^{\pm1},\theta\big)E^U(\Theta)\varphi,
\varphi\big\rangle_{\H_0}\\
&=\Big(\tfrac\d{\d\theta}\;\!\big\langle E^{U_0}\big((0,\theta]\big)
E^{U_0}(\Theta_0)\varphi_0,\varphi_0\big\rangle_{\H_0}\Big)
\Big(\tfrac\d{\d\theta}\;\!\big\langle E^U\big((0,\theta]\big)E^U(\Theta)\varphi,
\varphi\big\rangle_\H\Big)\\
&=\chi_{\Theta_0\cap\Theta}(\theta)\Big(\tfrac\d{\d\theta}\;\!
\big\langle E^{U_0}\big((0,\theta]\big)\varphi_0,\varphi_0\big\rangle_{\H_0}\Big)
\Big(\tfrac\d{\d\theta}\;\!\big\langle E^U\big((0,\theta]\big)\varphi,
\varphi\big\rangle_\H\Big)
\end{align*}
which proves \eqref{eq_null}. Using \eqref{eq_null}, we can then establish
\eqref{eq_intersection} as in the proof of \cite[Lemma~5.2.1]{Yaf92}.

(b) The singular component of the vector $\varphi$ relative to the unitary operator
$U$ can be written as $E^U(\Theta_{\rm s})\varphi$ for some Borel set
$\Theta_{\rm s}\subset[0,2\pi)$ of Lebesgue measure $0$, and similarly for $\varphi_0$
and $U_0$. Therefore, it follows from \eqref{eq_null}-\eqref{eq_intersection} that the
existence of the limit $a_\pm$ for the pair $(\varphi_0,\varphi)$ is equivalent to the
existence of $a_\pm$ for any of the pairs $(P_{\rm ac}(U_0)\varphi_0,\varphi\big)$,
$\big(\varphi_0,P_{\rm ac}(U)\varphi\big)$,
$\big(P_{\rm ac}(U_0)\varphi_0,P_{\rm ac}(U)\varphi\big)$, and for all four pairs the
value of $a_\pm$ is the same for a.e. $\theta\in[0,2\pi)$. 

(c) By definition of $w_\pm(U,U_0,J,\varepsilon)$, we have for
$\psi_0:=E^{U_0}(\Theta_0)P_{\rm ac}(U_0)\varphi_0$ and $\psi:=E^U(\Theta)\varphi$
\begin{equation}\label{eq_exchange}
\lim_{\varepsilon\searrow0}\big\langle w_\pm(U,U_0,J,\varepsilon)\psi_0,
\psi\big\rangle_\H
=\pm\lim_{\varepsilon\searrow0}\int_0^{2\pi}\d\theta\,g_\pm(\varepsilon)\;\!
\big\langle JR_0\big((1-\varepsilon)^{\pm1}\e^{i\theta}\big)\psi_0,
R\big((1-\varepsilon)^{\pm1}\e^{i\theta}\big)\psi\big\rangle_\H,
\end{equation}
and we know from \eqref{eq_intersection}-\eqref{eq_ac} that
$$
a_\pm\big(\psi_0,\psi,\theta\big)
=\chi_{\Theta_0\cap\Theta}(\theta)\;\!a_\pm(\varphi_0,\varphi,\theta).
$$
Therefore, to prove the claim, it is sufficient to show that we can exchange the limit
and the integral in \eqref{eq_exchange}. This will follow from Vitali's theorem
\cite[p.133]{Rud87} if we show that the family of functions
$\{f_{\pm,\varepsilon}\}_{\varepsilon\in(0,1)}\subset\lone([0,2\pi))$ given by
$$
f_{\pm,\varepsilon}(\theta):=g_\pm(\varepsilon)\;\!
\big\langle JR_0\big((1-\varepsilon)^{\pm1}\e^{i\theta}\big)\psi_0,
R\big((1-\varepsilon)^{\pm1}\e^{i\theta}\big)\psi\big\rangle_\H,
$$
is uniformly integrable on $[0,2\pi)$. For this, we first observe that a calculation
as in \eqref{eq_bound_w} shows for any Borel set ${\cal B}\subset[0,2\pi)$ that
$$
\left|\int_{\cal B}\d\theta\,f_{\pm,\varepsilon}(\theta)\right|
\le\|J\|_{\B(\H_0,\H)}\|\psi\|_\H\left(\int_{\cal B}\d\theta\left|\big\langle
\delta_0\big((1-\varepsilon)^{\pm1},\theta\big)\psi_0,
\psi_0\big\rangle_{\H_0}\right|\right)^{1/2}.
$$
Due to this bound and \eqref{eq_minus}, it is sufficient to show the uniform
integrability of the family of functions
$
\big\{\langle\delta_0(1-\varepsilon,\;\!\cdot\;\!)\psi_0,
\psi_0\rangle_{\H_0}\big\}_{\varepsilon\in(0,1)}
$.
Using the spectral decomposition of $U_0$ and Fubini's theorem we get
$$
\big\langle\delta_0(1-\varepsilon,\;\!\cdot\;\!)\psi_0,\psi_0\big\rangle_{\H_0}
=\tfrac1{2\pi}\int_0^{2\pi}\big(1-(1-\varepsilon)^2\big)
\big|1-(1-\varepsilon)\e^{i(\theta-\theta')}\big|^{-2}
\big\langle E^{U_0}(\d\theta')\psi_0,\psi_0\big\rangle_{\H_0}.
$$
Since the integrand is nothing but the Poisson integral kernel and
$\langle E^{U_0}(\d\theta')\psi_0,\psi_0\rangle_{\H_0}$ is a complex Baire measure, it
is possible to apply \cite[2nd~thm.~p.33]{Hof_1962} to infer that the function
$$
\C\setminus\S^1\ni r\e^{i\theta}\mapsto
\big\langle\delta_0(r,\theta)\psi_0,\psi_0\big\rangle_{\H_0}
$$
is harmonic on the open disc. Moreover, the integrals
$
\int_0^{2\pi}\d\theta\,
\big\langle\delta_0(r,\theta)\psi_0,\psi_0\big\rangle_{\H_0}
$
are bounded as $r\nearrow1$ due to \eqref{eq_int_delta}. Hence we can apply
\eqref{eq_derivative} and corollaries of  Fatou's theorem \cite[p.38]{Hof_1962} to
obtain that
$$
\lim_{\varepsilon\searrow0}\left\|\big\langle\delta_0(1-\varepsilon,\;\!\cdot\;\!)
\psi_0,\psi_0\big\rangle_{\H_0}\right\|_{\lone([0,2\pi))}
=\left\|\theta\mapsto\tfrac\d{\d\theta}\;\!\big\langle
P_{\rm ac}(U_0)E^{U_0}\big((0,\theta]\big)\psi_0,\psi_0\big\rangle_{\H_0}
\right\|_{\lone([0,2\pi))}.
$$
In consequence, there exists $\varepsilon_0>0$ such that the family
$
\big\{\langle\delta_0(1-\varepsilon,\;\!\cdot\;\!)\psi_0,
\psi_0\rangle_{\H_0}\big\}_{\varepsilon\in(0,\varepsilon_0]}
$
is uniformly integrable on $[0,2\pi)$. To conclude the proof, we note that the formula
\eqref{eq_norm} for the norm of $\delta_0(1-\varepsilon,\theta)$ implies that the
family
$
\big\{\langle\delta_0(1-\varepsilon,\;\!\cdot\;\!)\psi_0,
\psi_0\rangle_{\H_0}\big\}_{\varepsilon\in(\varepsilon_0,1)}
$
is also uniformly integrable on $[0,2\pi)$.
\end{proof}

\begin{Corollary}\label{cor_stat_wo}
Let $\D_0\subset\H_0$ and $\D\subset\H$ be dense sets, and assume that for each
$\varphi_0\in\D_0$ and $\varphi\in\D$ the limits $a_\pm(\varphi_0,\varphi,\theta)$
exist for a.e. $\theta\in[0,2\pi)$. Then, the following weak limits exist
$$
w_\pm(U,U_0,J):=\wlim_{\varepsilon\searrow0}w_\pm(U,U_0,J,\varepsilon)P_{\rm ac}(U_0).
$$
\end{Corollary}

\begin{proof}
The claim follows from Lemma \ref{lemma_limits}(c) applied with
$\Theta_0=\Theta=[0,2\pi)$ and a density argument taking into account the uniform
bound \eqref{eq_bound_uniform}.
\end{proof}

The weak limits $w_\pm(U,U_0,J)$ of Corollary \ref{cor_stat_wo} are called the
stationary wave operators for the triple $(U,U_0,J)$. When they exist, it can be shown
as in the self-adjoint case \cite[p.158-159]{Yaf92} that they posses the usual
properties of wave operators, namely, the inclusions
\begin{equation}\label{eq_inclusions}
\H_{\rm s}(U_0)\subset\ker\big(w_\pm(U,U_0,J)\big)
\quad\hbox{and}\quad
\Ran\big(w_\pm(U,U_0,J)\big)\subset\H_{\rm ac}(U)
\end{equation}
with $\H_{\rm s}(U_0)\subset\H_0$ the singular subspace of $U_0$, and the
intertwinning relation
\begin{equation}\label{eq_inter}
w_\pm(U,U_0,J)E^{U_0}(\Theta)=E^U(\Theta)w_\pm(U,U_0,J),
\quad\hbox{$\Theta\subset[0,2\pi)$ Borel set.}
\end{equation}

In the next theorem, we give more explicit conditions guaranteeing the existence of
$w_\pm(U,U_0,J)$, as well as representation formulas for $w_\pm(U,U_0,J)$. From now
on, we assume that there exist an auxiliary separable Hilbert space $\G$ and operators
$G_0\in\B(\H_0,\G)$, $G\in\B(\H,\G)$ such that $V=G^*G_0$.

\begin{Theorem}[Stationary wave operators]\label{thm_stat_wave}
Assume that for each $\varphi_0$ in a dense set $\D_0\subset\H_0$
\begin{equation}\label{eq_free_lim}
\slim_{\varepsilon\searrow0}
G_0U_0^*R_0\big((1-\varepsilon)^{\pm1}\e^{i\theta}\big)\varphi_0
\hbox{~exist for a.e. $\theta\in[0,2\pi)$,}
\end{equation}
and suppose that $G$ is weakly $U$-smooth. Then, the stationary wave operators
$w_\pm(U,U_0,J)$ exist and satisfy the representation formulas
\begin{equation}\label{eq_rep_stat}
\big\langle w_\pm(U,U_0,J)\varphi_0,\varphi\big\rangle_\H
=\int_0^{2\pi}\d\theta\,a_\pm(\varphi_0,\varphi,\theta),
\quad\varphi_0\in\D_0,~\varphi\in\H,
\end{equation}
with $a_\pm(\varphi_0,\varphi,\theta)$ given by \eqref{del_a_pm}.
\end{Theorem}

\begin{Remark}\label{rem_ass_equiv}
(a) The assumption \eqref{eq_free_lim} is equivalent to the condition
\begin{equation}\label{eq_free_lim_bis}
\slim_{\varepsilon\searrow0}
G_0R_0\big((1-\varepsilon)^{\pm1}\e^{i\theta}\big)\varphi_0
\hbox{~exist for a.e. $\theta\in[0,2\pi)$.}
\end{equation}
Indeed, if \eqref{eq_free_lim} is satisfied, then it follows from the identity
$R_0(z)=1+zU_0^*R_0(z)$ ($z\in\C\setminus\S^1$) that
\begin{align*}
\slim_{\varepsilon\searrow0}
G_0R_0\big((1-\varepsilon)^{\pm1}\e^{i\theta}\big)\varphi_0
=G_0\varphi_0+\e^{i\theta}\slim_{\varepsilon\searrow0}
G_0U_0^*R_0\big((1-\varepsilon)^{\pm1}\e^{i\theta}\big)\varphi_0,
\quad\hbox{a.e. $\theta\in[0,2\pi)$.}
\end{align*}
Alternatively, if \eqref{eq_free_lim_bis} is satisfied, then it follows from the
identity $U_0^*R_0(z)=z^{-1}\big(R_0(z)-1\big)$ ($z\in\C\setminus\S^1$) that
$$
\slim_{\varepsilon\searrow0}
G_0U_0^*R_0\big((1-\varepsilon)^{\pm1}\e^{i\theta}\big)\varphi_0
=\e^{-i\theta}\slim_{\varepsilon\searrow0}
G_0R_0\big((1-\varepsilon)^{\pm1}\e^{i\theta}\big)\varphi_0-\e^{-i\theta}G_0\varphi_0,
\quad\hbox{a.e. $\theta\in[0,2\pi)$.}
$$

(b) The roles of the triples $(\H_0,U_0,G_0)$ and $(\H,U,G)$ can be exchanged in the
statement of the theorem in order to obtain an alternative formulation of Theorem
\ref{thm_stat_wave} (see \cite[Lemma~5.2.6']{Yaf92} in the self-adjoint case).
\end{Remark}

\begin{proof}[Proof of Theorem \ref{thm_stat_wave}]
Using \eqref{eq_minus} and \eqref{eq_factors_bis}, we get for $\varepsilon\in(0,1)$
and $\theta\in[0,2\pi)$
\begin{align*}
&\pm g_\pm(\varepsilon)R\big((1-\varepsilon)^{\pm1}\e^{i\theta}\big)^*J
R_0\big((1-\varepsilon)^{\pm1}\e^{i\theta}\big)\\
&=\pm g_\pm(\varepsilon)R\big((1-\varepsilon)^{\pm1}\e^{i\theta}\big)^*
R\big((1-\varepsilon)^{\pm1}\e^{i\theta}\big)J\\
&\quad\mp g_\pm(\varepsilon)(1-\varepsilon)^{\pm1}\e^{i\theta}
\big(GUR\big((1-\varepsilon)^{\pm1}\e^{i\theta}\big)
R\big((1-\varepsilon)^{\pm1}\e^{i\theta}\big)^*\big)^*
G_0U_0^*R_0\big((1-\varepsilon)^{\pm1}\e^{i\theta}\big)\\
&=\delta(1-\varepsilon,\theta)J-(1-\varepsilon)^{\pm1}\e^{i\theta}
\big(GU\;\!\delta(1-\varepsilon,\theta)\big)^*
G_0U_0^*R_0\big((1-\varepsilon)^{\pm1}\e^{i\theta}\big).
\end{align*}
So, we obtain for $\varphi_0\in\D_0$ and $\varphi\in\H$
\begin{align}
&\pm g_\pm(\varepsilon)\;\!
\big\langle JR_0\big((1-\varepsilon)^{\pm1}\e^{i\theta}\big)\varphi_0,
R\big((1-\varepsilon)^{\pm1}\e^{i\theta}\big)\varphi\big\rangle_\H\nonumber\\
&=\big\langle\delta(1-\varepsilon,\theta)J\;\!\varphi_0,\varphi\big\rangle_\H
-(1-\varepsilon)^{\pm1}\e^{i\theta}\big\langle G_0U_0^*
R_0\big((1-\varepsilon)^{\pm1}\e^{i\theta}\big)\varphi_0,
GU\;\!\delta(1-\varepsilon,\theta)\varphi\big\rangle_\G.\label{eq_two_terms}
\end{align}
Now, the first term in \eqref{eq_two_terms} has a limit as $\varepsilon\searrow0$ for
a.e. $\theta\in[0,2\pi)$ due to \eqref{eq_derivative}. For the second term, the strong
limits
$
\slim_{\varepsilon\searrow0}
G_0U_0^*R_0\big((1-\varepsilon)^{\pm1}\e^{i\theta}\big)\varphi_0
$
exist for a.e. $\theta\in[0,2\pi)$ due to \eqref{eq_free_lim}. Moreover, the operator
$GU$ is weakly $U$-smooth, since $G$ is weakly $U$-smooth (see \eqref{def_weak}).
Thus, Lemma \ref{lemma_weak}(b) implies that the weak limit
$\wlim_{\varepsilon\searrow0}GU\delta(1-\varepsilon,\theta)\varphi$ exists for a.e.
$\theta\in[0,2\pi)$. Therefore, the second term in \eqref{eq_two_terms} also has a
limit as $\varepsilon\searrow0$ for a.e. $\theta\in[0,2\pi)$. In consequence, the
limits $a_\pm(\varphi_0,\varphi,\theta)$ exist for a.e. $\theta\in[0,2\pi)$, and the
claim follows from Corollary \ref{cor_stat_wo}.
\end{proof}

For each $\varepsilon\in(0,1)$, we now define auxiliary operators
$w_\pm(U_0,U_0,J^*J,\varepsilon)\in\B(\H_0)$ in terms of the sesquilinear form
\begin{align*}
&\big\langle w_\pm(U_0,U_0,J^*J,\varepsilon)\varphi_0,\psi_0\big\rangle_{\H_0}\\
&:=\pm g_\pm(\varepsilon)\int_0^{2\pi}\d\theta\,
\big\langle J^*JR_0\big((1-\varepsilon)^{\pm1}\e^{i\theta}\big)\varphi_0,
R_0\big((1-\varepsilon)^{\pm1}\e^{i\theta}\big)\psi_0\big\rangle_{\H_0},
\quad\varphi_0,\psi_0\in\H_0.
\end{align*}
Using Theorem \ref{thm_stat_wave}, we can show the existence of the corresponding
stationary wave operators
$$
w_\pm(U_0,U_0,J^*J)
:=\wlim_{\varepsilon\searrow0}w_\pm(U_0,U_0,J^*J,\varepsilon)P_{\rm ac}(U_0),
$$
and relate them with the initial stationary wave operators $w_\pm(U,U_0,J)$\;\!:

\begin{Corollary}\label{cor_id_stat}
Suppose that the assumptions of Theorem \ref{thm_stat_wave} are satisfied for a dense
set $\D_0\subset\H_0$, and let $\varphi_0\in\D_0$.
\begin{enumerate}
\item[(a)] For a.e. $\theta\in[0,2\pi)$ the following weak limits exist and are equal
\begin{align}
&\wlim_{\varepsilon\searrow0}
G\delta(1-\varepsilon,\theta)w_\pm(U,U_0,J)U_0\varphi_0\nonumber\\
&=\pm\wlim_{\varepsilon\searrow0}g_\pm(\varepsilon)G
R\big((1-\varepsilon)^{\pm1}\e^{i\theta}\big)^*UJ
R_0\big((1-\varepsilon)^{\pm1}\e^{i\theta}\big)\varphi_0.\label{eq_limits}
\end{align}
\item[(b)] The stationary wave operators
$$
w_\pm(U_0,U_0,J^*J)
:=\wlim_{\varepsilon\searrow0}w_\pm(U_0,U_0,J^*J,\varepsilon)P_{\rm ac}(U_0)
$$
exist and satisfy the chain rule
\begin{equation}\label{eq_product}
w_\pm(U,U_0,J)^*\;\!w_\pm(U,U_0,J)=w_\pm(U_0,U_0,J^*J).
\end{equation}
\end{enumerate}
\end{Corollary}

\begin{proof}
(a) Let $\Theta\subset[0,2\pi)$ be a Borel set and $\varphi\in\H$. Since
$w_\pm(U,U_0,J)\varphi_0\in\H_{\rm ac}(U)$, it follows from \eqref{eq_derivative} that
\begin{align*}
\big\langle E^U(\Theta)w_\pm(U,U_0,J)\varphi_0,\varphi\big\rangle_\H
&=\int_\Theta\d\theta\,\tfrac\d{\d\theta}
\big\langle E^U\big((0,\theta]\big)w_\pm(U,U_0,J)\varphi_0,\varphi\big\rangle_\H\\
&=\int_\Theta\d\theta\,\lim_{\varepsilon\searrow0}
\big\langle\delta(1-\varepsilon,\theta)w_\pm(U,U_0,J)\varphi_0,\varphi\big\rangle_\H.
\end{align*}
On the other hand, we know from \eqref{eq_intersection} and \eqref{eq_rep_stat} that
$$
\big\langle w_\pm(U,U_0,J)\varphi_0,E^U(\Theta)\varphi\big\rangle_\H
=\pm\int_\Theta\d\theta\,\lim_{\varepsilon\searrow0}g_\pm(\varepsilon)\;\!
\big\langle R\big((1-\varepsilon)^{\pm1}\e^{i\theta}\big)^*J
R_0\big((1-\varepsilon)^{\pm1}\e^{i\theta}\big)\varphi_0,\varphi\big\rangle_\H.
$$
Since $\Theta$ is arbitrary, we thus obtain by comparison that
\begin{align}
&\lim_{\varepsilon\searrow0}\big\langle \delta(1-\varepsilon,\theta)w_\pm(U,U_0,J)
\varphi_0,\varphi\big\rangle_\H\nonumber\\
&=\pm\lim_{\varepsilon\searrow0}
\big\langle g_\pm(\varepsilon)R\big((1-\varepsilon)^{\pm1}\e^{i\theta}\big)^*J
R_0\big((1-\varepsilon)^{\pm1}\e^{i\theta}\big)\varphi_0,\varphi\big\rangle_\H,
\quad\hbox{a.e. $\theta\in[0,2\pi)$.}\label{eq_limit_v1}
\end{align}
Setting $\varphi=U^*G^*\zeta$ with $\zeta\in\G$ and using \eqref{eq_inter}, we get
\begin{align}
&\lim_{\varepsilon\searrow0}\big\langle G\delta(1-\varepsilon,\theta)w_\pm(U,U_0,J)
U_0\varphi_0,\zeta\big\rangle_\G\nonumber\\
&=\pm\lim_{\varepsilon\searrow0}
\big\langle g_\pm(\varepsilon)GR\big((1-\varepsilon)^{\pm1}\e^{i\theta}\big)^*UJ
R_0\big((1-\varepsilon)^{\pm1}\e^{i\theta}\big)\varphi_0,\zeta\big\rangle_\G,
\quad\hbox{a.e. $\theta\in[0,2\pi)$.}\label{eq_limit_v2}
\end{align}
In consequence, it only remains to show the existence of the limits in
\eqref{eq_limits}.

The existence of the first limit in \eqref{eq_limits} follows from Lemma
\ref{lemma_weak}(b) for the pair $(U,G)$ and $\Theta=[0,2\pi)$. For the second limit,
we know that the limit \eqref{eq_limit_v2} exists for a.e. $\theta\in[0,2\pi)$, and
the corresponding set of full measure may be assumed to be independent of the choice
of $\zeta$ in the set of linear combinations of some basis in $\G$. Therefore, to
prove the existence the second limit in \eqref{eq_limits} it is sufficient to show
(see \cite[Lemma~1.31]{Kat95}) the existence of a function
$c_{\varphi_0}:[0,2\pi)\to[0,\infty)$ such that
\begin{equation}\label{eq_the_bound}
|g_\pm(\varepsilon)|\;\!\big\|GR\big((1-\varepsilon)^{\pm1}\e^{i\theta}\big)^*UJ
R_0\big((1-\varepsilon)^{\pm1}\e^{i\theta}\big)\varphi_0\big\|_\G
\le c_{\varphi_0}(\theta),\quad\hbox{$\varepsilon\in(0,1)$, a.e. $\theta\in[0,2\pi)$.}
\end{equation}
But the l.h.s. of \eqref{eq_the_bound} satisfies the bound
$$
\|UJ\|_{\B(\H_0,\H)}\left(|g_\pm(\varepsilon)|^{1/2}\;\!
\big\|GR\big((1-\varepsilon)^{\pm1}\e^{i\theta}\big)^*\big\|_{\B(\H,\G)}\right)
\left(|g_\pm(\varepsilon)|^{1/2}\;\!
\big\|R_0\big((1-\varepsilon)^{\pm1}\e^{i\theta}\big)\varphi_0\big\|_{\H_0}\right).
$$
So, \eqref{eq_the_bound} follows from \eqref{eq_bound_res} and an application of Lemma
\ref{lemma_weak}(a) for the pair $(U,G)$ and $\Theta=[0,2\pi)$.

(b) Setting $\varphi=w_\pm(U,U_0,J)\psi_0$ with $\psi_0\in\D_0$ in
\eqref{eq_rep_stat}, and using \eqref{eq_norm}, \eqref{eq_factors_bis} and
\eqref{eq_inter}, we obtain
\begin{align*}
&\big\langle w_\pm(U,U_0,J)\varphi_0,w_\pm(U,U_0,J)\psi_0\big\rangle_\H\\
&=\int_0^{2\pi}\d\theta\,a_\pm\big(\varphi_0,w_\pm(U,U_0,J)\psi_0,\theta\big)\\
&=\pm\int_0^{2\pi}\d\theta\,\lim_{\varepsilon\searrow0}g_\pm(\varepsilon)\;\!
\left\langle JR_0\big((1-\varepsilon)^{\pm1}\e^{i\theta}\big)\varphi_0,
R\big((1-\varepsilon)^{\pm1}\e^{i\theta}\big)w_\pm(U,U_0,J)\psi_0\right\rangle_\H\\
&=\int_0^{2\pi}\d\theta\,\lim_{\varepsilon\searrow0}
\big\langle J\varphi_0,\delta(1-\varepsilon,\theta)w_\pm(U,U_0,J)\psi_0\big\rangle_\H\\
&\quad-\int_0^{2\pi}\d\theta\,\lim_{\varepsilon\searrow0}
\big\langle(1-\varepsilon)^{\pm1}\e^{i\theta}
G_0U_0^*R_0\big((1-\varepsilon)^{\pm1}\e^{i\theta}\big)\varphi_0,
G\delta(1-\varepsilon,\theta)w_\pm(U,U_0,J)U_0\psi_0\big\rangle_\G.
\end{align*}
Now, using \eqref{eq_limit_v1} for the first term, using \eqref{eq_free_lim} and
\eqref{eq_limit_v2} for the second term, and then using \eqref{eq_factors_bis}, we
obtain that
\begin{align*}
&\big\langle w_\pm(U,U_0,J)\varphi_0,w_\pm(U,U_0,J)\psi_0\big\rangle_\H\\
&=\pm\int_0^{2\pi}\d\theta\,\lim_{\varepsilon\searrow0}g_\pm(\varepsilon)\;\!
\big\langle J\varphi_0,R\big((1-\varepsilon)^{\pm1}\e^{i\theta}\big)^*J
R_0\big((1-\varepsilon)^{\pm1}\e^{i\theta}\big)\psi_0\big\rangle_\H\\
&\quad\mp\int_0^{2\pi}\d\theta\,\lim_{\varepsilon\searrow0}g_\pm(\varepsilon)\;\!
\big\langle(1-\varepsilon)^{\pm1}\e^{i\theta}
G_0U_0^*R_0\big((1-\varepsilon)^{\pm1}\e^{i\theta}\big)\varphi_0,\\
&\quad\quad GR\big((1-\varepsilon)^{\pm1}\e^{i\theta}\big)^*UJ
R_0\big((1-\varepsilon)^{\pm1}\e^{i\theta}\big)\psi_0\big\rangle_\G\\
&=\pm\int_0^{2\pi}\d\theta\,\lim_{\varepsilon\searrow0}g_\pm(\varepsilon)\;\!
\big\langle J^*JR_0\big((1-\varepsilon)^{\pm1}\e^{i\theta}\big)\varphi_0,
R_0\big((1-\varepsilon)^{\pm1}\e^{i\theta}\big)\psi_0\big\rangle_{\H_0}.
\end{align*}
This, together with an application of Corollary \ref{cor_stat_wo} for the triple
$(U_0,U_0,J^*J)$, implies the claims.
\end{proof}

In our next theorem, we give explicit conditions guaranteeing the existence of the
strong wave operators
$$
W_\pm(U,U_0,J):=\slim_{n\to\pm\infty}U^nJU_0^{-n}P_{\rm ac}(U_0)\in\B(\H_0,\H),
$$
as well as their identity with the stationary wave operators $w_\pm(U,U_0,J)$. But
first, we need a last preparatory lemma\;\!:

\begin{Lemma}\label{lemma_strong_stat}
\begin{enumerate}
\item[(a)] The strong wave operators $W_\pm(U,U_0,J)$ exist if and only if the weak
wave operators
$$
\widetilde W_\pm(U,U_0,J)
:=\wlim_{n\to\pm\infty}P_{\rm ac}(U)U^nJU_0^{-n}P_{\rm ac}(U_0)
$$
and
$$
\widetilde W_\pm(U_0,U_0,J^*J)
:=\wlim_{n\to\pm\infty}P_{\rm ac}(U_0)U_0^nJ^*JU_0^{-n}P_{\rm ac}(U_0)
$$
exist and satisfy the chain rule
$$
\widetilde W_\pm(U,U_0,J)^*\;\!\widetilde W_\pm(U,U_0,J)
=\widetilde W_\pm(U_0,U_0,J^*J),
$$
in which case we have the identity $W_\pm(U,U_0,J)=\widetilde W_\pm(U,U_0,J)$.

\item[(b)] Let $\D_0\subset\H_0$ and $\D\subset\H$ be dense sets, and assume that for
each $\varphi_0\in\D_0$ and $\varphi\in\D$ the weak limits
$\wlim_{\varepsilon\searrow0}G_0\delta_0(1-\varepsilon,\theta)\varphi_0$ and
$\wlim_{\varepsilon\searrow0}G\delta(1-\varepsilon,\theta)\varphi$ exist for a.e.
$\theta\in[0,2\pi)$. Then, the weak wave operators $\widetilde W_\pm(U,U_0,J)$ exist.

\item[(c)] Suppose that the assumptions of Theorem \ref{thm_stat_wave} are satisfied
for a dense set $\D_0\subset\H_0$ and that the weak wave operators
$\widetilde W_\pm(U_0,U_0,J^*J)$ exist. Then, the strong wave operators
$W_\pm(U,U_0,J)$ exist and coincide with the stationary wave operators
$w_\pm(U,U_0,J)$.
\end{enumerate}
\end{Lemma}

\begin{proof}
(a) The claim can be proved as in \cite[Thm.~2.2.1]{Yaf92}.

(b) Set $W(n):=P_{\rm ac}(U)U^nJU_0^{-n}P_{\rm ac}(U_0)$ for each $n\in\N$, and let
$\D_0'\subset\H_{\rm ac}(U_0)$ (resp. $\D'\subset\H_{\rm ac}(U)$) be the set dense in
$\H_{\rm ac}(U_0)$ (resp. $\H_{\rm ac}(U)$) given by Lemma \ref{lemma_weak}(c). Since
$\|W(n)\|_{\B(\H_0,\H)}\le\|J\|_{\B(\H_0,\H)}$ for all $n\in\Z$, it is sufficient to
show that
\begin{equation}\label{eq_needed_lim}
\lim_{n\to\pm\infty}\big\langle W(n)\psi_0,\psi\big\rangle_\H
=\lim_{n\to\pm\infty}\big\langle U^nJU_0^{-n}\psi_0,\psi\big\rangle_\H
\end{equation}
exist for all $\psi_0\in\D_0'$ and $\psi\in\D'$. Now, for any $n_1,n_2\in\Z$ with
$n_1<n_2$, a telescoping summation gives
\begin{align*}
U^{n_2}JU_0^{-n_2}-U^{n_1}JU_0^{-n_1}
&=\sum_{n=n_1+1}^{n_2}\big(U^nJU_0^{-n}-U^{(n-1)}JU_0^{-(n-1)}\big)\\
&=\sum_{n=n_1+1}^{n_2}U^{(n-1)}(UJ-JU_0)U_0^{-n}\\
&=-\sum_{n=n_1+1}^{n_2}\big(GU^{-(n-1)}\big)^*G_0U_0^{-n}.
\end{align*}
Therefore, using Cauchy-Schwartz inequality we obtain the estimate
\begin{align*}
\Big|\big\langle W(n_2)\psi_0,\psi\big\rangle_\H
-\big\langle W(n_2)\psi_0,\psi\big\rangle_\H\Big|
&=\left|\left\langle\sum_{n=n_1+1}^{n_2}\big(GU^{-(n-1)}\big)^*G_0U_0^{-n}\psi_0,
\psi\right\rangle_\H\right|\\
&\le\left(\sum_{n=n_1+1}^{n_2}\big\|G_0U_0^{-n}\psi_0\big\|_\G^2\right)^{1/2}
\left(\sum_{n=n_1+1}^{n_2}\big\|GU^{-(n-1)}\psi\big\|_\G^2\right)^{1/2}.
\end{align*}
Then, we infer from Lemma \ref{lemma_weak}(c) that this expression tends to zero as
$n_1,n_2\to\infty$, which implies the existence of the limit \eqref{eq_needed_lim}.

(c) The stationary wave operators $w_\pm(U_0,U_0,J^*J)$ exist due to Corollary
\ref{cor_id_stat}(b). Furthermore, a calculation using \eqref{eq_series} and
\eqref{eq_inclusions} gives for $\varphi_0\in\D_0$ and $\psi\in\H_0$
\begin{align*}
&\big\langle w_+(U_0,U_0,J^*J)\varphi_0,\psi_0\big\rangle_{\H_0}\\
&=\tfrac1{2\pi}\lim_{\varepsilon\searrow0}\big(1-(1-\varepsilon)^2\big)
\sum_{n,m\ge0}(1-\varepsilon)^{n+m}\int_0^{2\pi}\d\theta\,\e^{i\theta(n-m)}
\big\langle P_{\rm ac}(U_0)U_0^mJ^*JU_0^{-n}P_{\rm ac}(U_0)\varphi_0,
\psi_0\big\rangle_{\H_0}\\
&=\lim_{\varepsilon\searrow0}\big(1-(1-\varepsilon)^2\big)
\sum_{n\ge0}(1-\varepsilon)^{2n}\;\!\big\langle P_{\rm ac}(U_0)U_0^nJ^*JU_0^{-n}
P_{\rm ac}(U_0)\varphi_0,\psi_0\big\rangle_{\H_0},
\end{align*}
and similarly
\begin{align*}
&\big\langle w_-(U_0,U_0,J^*J)\varphi_0,\psi_0\big\rangle_{\H_0}\\
&=\lim_{\varepsilon\searrow0}\big(1-(1-\varepsilon)^2\big)
\sum_{n\ge1}(1-\varepsilon)^{2(n-1)}\;\!\big\langle P_{\rm ac}(U_0)U_0^{-n}J^*JU_0^n
P_{\rm ac}(U_0)\varphi_0,\psi_0\big\rangle_{\H_0}.
\end{align*}
Therefore, by appying a general Tauberian theorem as in the in the self-adjoint case
(see \cite[p.76 \& 94]{Yaf92}), one gets that
\begin{align*}
\big\langle w_\pm(U_0,U_0,J^*J)\varphi_0,\psi_0\big\rangle_{\H_0}
&=\lim_{n\to\pm\infty}\big\langle P_{\rm ac}(U_0)U_0^nJ^*JU_0^{-n}
P_{\rm ac}(U_0)\varphi_0,\psi_0\big\rangle_{\H_0}\\
&=\big\langle\widetilde W_\pm(U_0,U_0,J^*J)\varphi_0,\psi_0\big\rangle_{\H_0},
\end{align*}
which shows that the wave operators $w_\pm(U_0,U_0,J^*J)$ and
$\widetilde W_\pm(U_0,U_0,J^*J)$ coincide.

Now, since $G$ is weakly $U$-smooth, it follows from Lemma \ref{lemma_weak}(b) for the
pair $(U,G)$ and $\Theta=[0,2\pi)$ that for each $\varphi\in\H$ the limit
$\wlim_{\varepsilon\searrow0}G\delta(1-\varepsilon,\theta)\varphi$ exists for a.e.
$\theta\in[0,2\pi)$. Furthermore, Remark \ref{rem_ass_equiv}(a) implies that for each
$\varphi_0\in\D_0$ the limit
$\wlim_{\varepsilon\searrow0}G_0\delta_0(1-\varepsilon,\theta)\varphi_0$ exists for
a.e. $\theta\in[0,2\pi)$. So, it follows from point (b) that the weak wave operators
$\widetilde W_\pm(U,U_0,J)$ exist, and we can show as above that they coincide with
the stationary wave operators $w_\pm(U,U_0,J)$ (which exist by Theorem
\ref{thm_stat_wave}). In particular, it follows from \eqref{eq_product} that
$$
\widetilde W_\pm(U,U_0,J)^*\;\!\widetilde W_\pm(U,U_0,J)
=\widetilde W_\pm(U_0,U_0,J^*J),
$$
and then the claim follows directly from point (a).
\end{proof}

We are finally in position to present the main theorem of this section on the
existence of the strong wave operators and their identity with the stationary wave
operators. We use the notation
$$
B(z):=GR(z)G^*\in\B(\G),\quad z\in\C\setminus\S^1.
$$

\begin{Theorem}[Strong wave operators]\label{thm_strong_wave}
Assume that \eqref{eq_free_lim} is satisfied for a dense set $\D_0\subset\H_0$ and
that
\begin{equation}\label{eq_B_limit}
B_\pm(\theta)
:=\wlim_{\varepsilon\searrow0}B\big((1-\varepsilon)^{\pm1}\e^{i\theta}\big)
\hbox{~exist for a.e. $\theta\in[0,2\pi)$.}
\end{equation}
Then, the strong wave operators $W_\pm(U,U_0,J)$ exist and coincide with the
stationary wave operators $w_\pm(U,U_0,J)$.
\end{Theorem}

\begin{proof}
The assumption \eqref{eq_B_limit} implies that $G$ is weakly $U$-smooth. Therefore, by
Lemma \ref{lemma_strong_stat}(c), it is only necessary to show the existence of the
weak wave operators $\widetilde W_\pm(U_0,U_0,J^*J)$. We do this by checking the
assumptions of Lemma \ref{lemma_strong_stat}(b) for the triple $(U_0,U_0,J^*J)$. Since
\begin{align*}
(J^*J)U_0-U_0(J^*J)
=J^*V-V^*J
=(GJ)^*G_0-G_0^*(GJ),
\end{align*}
it is sufficient to verify for $\varphi_0\in\D_0$ that the weak limits
$$
\wlim_{\varepsilon\searrow0}G_0\delta_0(1-\varepsilon,\theta)\varphi_0
\quad\hbox{and}\quad
\wlim_{\varepsilon\searrow0}GJ\;\!\delta_0(1-\varepsilon,\theta)\varphi_0
$$
exist for a.e. $\theta\in[0,2\pi)$. The existence of the first limit follows from
Remark \ref{rem_ass_equiv}(a). To establish the existence of the second limit, it is
sufficient to show (see the proof of \cite[Lemma~5.1.2]{Yaf92}) that there exists a
function $c_{\varphi_0}:[0,2\pi)\to[0,\infty)$ such that
\begin{equation}\label{eq_bound_GJ}
\big\|GJ\;\!\delta_0(1-\varepsilon,\theta)\varphi_0\big\|_\G\le c_{\varphi_0}(\theta),
\quad\hbox{$\varepsilon\in(0,1]$, a.e. $\theta\in[0,2\pi)$.}
\end{equation}
Multiplying \eqref{eq_factors_bis} by $G$ on the left gives
\begin{equation}\label{eq_GJ}
GJR_0(z)=GR(z)J+B({\bar z}^{-1})^*G_0U_0^*R_0(z),\quad z\in\C\setminus\S^1.
\end{equation}
Setting $z=(1-\varepsilon)\e^{i\theta}$ and applying on the vector
$R_0\big((1-\varepsilon)\e^{i\theta}\big)^*\varphi_0$, we get
\begin{align*}
GJ\;\!\delta_0(1-\varepsilon,\theta)\varphi_0
&=g_+(\varepsilon)GR\big((1-\varepsilon)\e^{i\theta}\big)J
R_0\big((1-\varepsilon)\e^{i\theta}\big)^*\varphi_0\\
&\quad+B\big((1-\varepsilon)^{-1}\e^{i\theta}\big)^*
G_0U_0^*\delta_0(1-\varepsilon,\theta)\varphi_0,
\end{align*}
which implies that
\begin{align*}
&\big\|GJ\;\!\delta_0(1-\varepsilon,\theta)\varphi_0\big\|_\G\\
&\le\|J\|_{\B(\H_0,\H)}\left(g_+(\varepsilon)^{1/2}
\big\|GR\big((1-\varepsilon)\e^{i\theta}\big)\big\|_{\B(\H,\G)}\right)
\left(g_+(\varepsilon)^{1/2}
\big\|R_0\big((1-\varepsilon)\e^{i\theta}\big)\varphi_0\big\|_{\H_0}\right)\\
&\quad+\big\|B\big((1-\varepsilon)^{-1}\e^{i\theta}\big)\big\|_{\B(\G)}
\big\|G_0U_0^*\delta_0(1-\varepsilon,\theta)\varphi_0\big\|_\G.
\end{align*}
This inequality, together with \eqref{eq_bound_res}, Lemma \ref{lemma_weak}(a) for the
pair $(U,G)$ and $\Theta=[0,2\pi)$, and the assumptions \eqref{eq_free_lim} and
\eqref{eq_B_limit}, implies the bound \eqref{eq_bound_GJ} needed to conclude the
proof.
\end{proof}

\begin{Example}[Trace class perturbation]\label{ex_trace_class}
The assumptions of Theorem \ref{thm_strong_wave} are satisfied for the set $\D_0=\H_0$
when the perturbation $V$ is trace class, or equivalently when the operators $G_0$ and
$G$ are Hilbert-Schmidt. Indeed, if $G_0\in S_2(\H_0,\G)$, then \eqref{eq_free_lim} is
satisfied for the set $\D_0=\H_0$ due to Lemma \ref{lemma_HS}(c) and Remark
\ref{rem_ass_equiv}(a), and if $G\in S_2(\H,\G)$, then \eqref{eq_B_limit} follows from
Lemma \ref{lemma_HS}(b).
\end{Example}

%--------------------------------------------------------------------------------------
\section{Representation formulas for the scattering matrix}\label{section_matrix}
\setcounter{equation}{0}
%--------------------------------------------------------------------------------------

In this section, we define the scattering operator and the scattering matrix for the
triple $(U,U_0,J)$, and we derive representation formulas for the scattering matrix.

If the strong wave operators $W_\pm(U,U_0,J)$ exist, then the scattering operator for
the triple $(U,U_0,J)$ is defined as
$$
S(U,U_0,J):=W_+(U,U_0,J)^*\;\!W_-(U,U_0,J)\in\B(\H_0).
$$
As in the self-adjoint case \cite[Sec.~2.4]{Yaf92}, it is easily seen that the
operator $S(U,U_0,J)$ vanishes on $\H_{\rm sc}(U_0)$ and has range in
$\H_{\rm ac}(U_0)$. Furthermore, if the operators $W_\pm(U,U_0,J)$ are isometric on
$\H_{\rm ac}(U_0)$, then $S(U,U_0,J)$ is isometric on $\H_{\rm ac}(U_0)$ if and only
if $\Ran\big(W_-(U,U_0,J)\big)\subset\Ran\big(W_+(U,U_0,J)\big)$, and $S(U,U_0,J)$ is
unitary on $\H_{\rm ac}(U_0)$ if and only if
$\Ran\big(W_-(U,U_0,J)\big)=\Ran\big(W_+(U,U_0,J)\big)$.

Since $\H_{\rm ac}(U_0)$ is a reducing subspace for $S(U,U_0,J)$ and the restriction
$$
S^{\rm(ac)}(U,U_0,J):=S(U,U_0,J)\upharpoonright\H_{\rm ac}(U_0)
$$
commutes with $U_0^{\rm(ac)}$, the operator $S^{\rm(ac)}(U,U_0,J)$ decomposes in the
Hilbert space $\int_{\widehat\sigma_0}^\oplus\d\theta\,\h_0(\theta)$ (see
\eqref{def_F_0} and \cite[Thm.~7.2.3(b)]{BS87}). Namely, there exist for a.e.
$\theta\in\widehat\sigma_0$ operators $S(\theta)\in\B\big(\h_0(\theta)\big)$ (we do
not write their dependency on $U$, $U_0$ and $J$) such that
$$
F_0^{\rm(ac)}S^{\rm(ac)}(U,U_0,J)\big(F_0^{\rm(ac)}\big)^*
=\int_{\widehat\sigma_0}^\oplus\d\theta\,S(\theta).
$$
The family of operators $S(\theta)$ is called the scattering matrix for the triple
$(U,U_0,J)$, and the properties of the scattering operator can be reformulated in
terms of the scattering matrix. For instance, if the operators $W_\pm(U,U_0,J)$ are
isometric on $\H_{\rm ac}(U_0)$, then $S(\theta)$ is isometric on $\h_0(\theta)$ for
a.e. $\theta\in\widehat\sigma_0$ if and only if
$\Ran\big(W_-(U,U_0,J)\big)\subset\Ran\big(W_+(U,U_0,J)\big)$, and $S(\theta)$ is
unitary on $\h_0(\theta)$ for a.e. $\theta\in\widehat\sigma_0$ if and only if
$\Ran\big(W_-(U,U_0,J)\big)=\Ran\big(W_+(U,U_0,J)\big)$.

Similarly, if the stationary wave operators $w_\pm(U_0,U_0,J^*J)$ exist, then
$\H_{\rm ac}(U_0)$ is a reducing subspace for $w_\pm(U_0,U_0,J^*J)$, and the
restriction
$$
w_\pm^{\rm(ac)}(U_0,U_0,J^*J):=w_\pm(U_0,U_0,J^*J)\upharpoonright\H_{\rm ac}(U_0)
$$
commutes with $U_0^{\rm(ac)}$. Thus, there exist for a.e. $\theta\in\widehat\sigma_0$
operators $u_\pm(\theta)\in\B\big(\h_0(\theta)\big)$ (we do not write their dependency
on $U_0$ and $J^*J$) such that
\begin{equation}\label{def_u_pm}
F_0^{\rm(ac)}w_\pm^{\rm(ac)}(U_0,U_0,J^*J)\big(F_0^{\rm(ac)}\big)^*
=\int_{\widehat\sigma_0}^\oplus\d\theta\,u_\pm(\theta).
\end{equation}
If the assumptions of Theorem \ref{thm_stat_wave} are satisfied (so that the
stationary wave operators $w_\pm(U,U_0,J)$ exist) and $w_\pm(U,U_0,J)$ are isometric
on $\H_{\rm ac}(U_0)$, then $w_\pm^{\rm(ac)}(U_0,U_0,J^*J)=1_{\H_{\rm ac}(U_0)}$ due
to \eqref{eq_product}. Thus, $u_\pm(\theta)=1_{\h_0(\theta)}$ for a.e.
$\theta\in\widehat\sigma_0$. This occurs for instance in the one-Hilbert space case
$\H_0=\H$ and $J=1_{\H_0}$.

The following lemma is a first step in the derivation of the representation formulas
for the scattering matrix.

\begin{Lemma}\label{lemma_scatt}
Let $\Theta\subset[0,2\pi)$ be a Borel set, assume that \eqref{eq_free_lim} is
satisfied for a dense set $\D_0\subset\H_0$, suppose that $G$ is weakly $U$-smooth,
and let $\varphi_0,\psi_0\in\D_0$.
\begin{enumerate}
\item[(a)] We have the equalities
\begin{align}
&\big\langle E^U(\Theta)w_-(U,U_0,J)\varphi_0,
w_+(U,U_0,J)\psi_0\big\rangle_\H\nonumber\\
&=\big\langle E^{U_0}(\Theta)w_\pm(U_0,U_0,J^*J)\varphi_0,
\psi_0\big\rangle_{\H_0}\nonumber\\
&\quad\pm2\pi\int_\Theta\d\theta\,\lim_{\varepsilon\searrow0}\big\langle
T_\pm\big((1-\varepsilon)\e^{i\theta}\big)\delta_0(1-\varepsilon,\theta)\varphi_0,
\delta_0(1-\varepsilon,\theta)\psi_0\big\rangle_{\H_0},\label{eq_pm}
\end{align}
with
$$
T_+(z):=U_0^*J^*V-V^*R(z)V\in\B(\H_0)
\quad\hbox{and}\quad
T_-(z):=\big(T_+({\bar z}^{-1})\big)^*\in\B(\H_0),
\quad z\in\C\setminus\S^1.
$$

\item[(b)] If in addition \eqref{eq_B_limit} is satisfied, then we have the equalities
\begin{align}
&\big\langle\big(S(\theta)-u_\pm(\theta)\big)(F_0\varphi_0)(\theta),
(F_0\psi_0)(\theta)\big\rangle_{\h_0(\theta)}\label{eq_pm_bis}\\
&=\pm2\pi\lim_{\varepsilon\searrow0}\big\langle P_{\rm ac}(U_0)
T_\pm\big((1-\varepsilon)\e^{i\theta}\big)P_{\rm ac}(U_0)\;\!
\delta_0(1-\varepsilon,\theta)\varphi_0,\delta_0(1-\varepsilon,\theta)
\psi_0\big\rangle_{\H_0},
\quad\hbox{a.e. $\theta\in\widehat\sigma_0$.}\nonumber
\end{align}
\end{enumerate}
\end{Lemma}

\begin{proof}
(a) The stationary wave operators $w_\pm(U,U_0,J)$ and $w_\pm(U_0,U_0,J^*J)$ exist due
to Theorem \ref{thm_stat_wave} and Corollary \ref{cor_id_stat}. Let
$\Theta\subset[0,2\pi)$ be a Borel set and take $\varphi_0,\psi_0\in\D_0$. Using
successively Lemma \ref{lemma_limits}(c) with $\varphi=w_+(U,U_0,J)\psi_0$, the
resolvent equation \eqref{eq_factors}, the definition of
$\delta(1-\varepsilon,\theta)$, and \eqref{eq_limit_v1}-\eqref{eq_limit_v2} we get
\begin{align}
&\big\langle E^U(\Theta)w_-(U,U_0,J)\varphi_0,
w_+(U,U_0,J)\psi_0\big\rangle_\H\nonumber\\
&=\int_\Theta\d\theta\lim_{\varepsilon\searrow0}g_+(\varepsilon)
\big\langle R\big((1-\varepsilon)\e^{i\theta}\big)w_-(U,U_0,J)\varphi_0,
JR_0\big((1-\varepsilon)\e^{i\theta}\big)\psi_0\big\rangle_\H\nonumber\\
&=\int_\Theta\d\theta\lim_{\varepsilon\searrow0}\big\langle
\delta(1-\varepsilon,\theta)w_-(U,U_0,J)\varphi_0,J\psi_0\big\rangle_\H\nonumber\\
&\quad-\int_\Theta\d\theta\lim_{\varepsilon\searrow0}\big\langle
(1-\varepsilon)\e^{-i\theta}R_0\big((1-\varepsilon)\e^{i\theta}\big)^*U_0V^*U
\delta(1-\varepsilon,\theta)w_-(U,U_0,J)\varphi_0,\psi_0\big\rangle_{\H_0}\nonumber\\
&=-\int_\Theta\d\theta\lim_{\varepsilon\searrow0}g_-(\varepsilon)
\big\langle R\big((1-\varepsilon)^{-1}\e^{i\theta}\big)^*J
R_0\big((1-\varepsilon)^{-1}\e^{i\theta}\big)\varphi_0,J\psi_0\big\rangle_\H
\label{eq_first_term}\\
&\quad+\int_\Theta\d\theta\lim_{\varepsilon\searrow0}g_-(\varepsilon)(1-\varepsilon)
\e^{-i\theta}\big\langle R_0\big((1-\varepsilon)\e^{i\theta}\big)^*
U_0V^*UR\big((1-\varepsilon)^{-1}\e^{i\theta}\big)^*J
R_0\big((1-\varepsilon)^{-1}\e^{i\theta}\big)\varphi_0,\psi_0\big\rangle_{\H_0}.
\nonumber
\end{align}
Using the formula \eqref{eq_factors} for the resolvent
$R\big((1-\varepsilon)^{-1}\e^{i\theta}\big)^*$ and then Lemma \ref{lemma_limits}(c)
for the triple $(U_0,U_0,J^*J)$, we obtain that the first term in
\eqref{eq_first_term} is equal to
\begin{align*}
&\big\langle E^{U_0}(\Theta)w_-(U_0,U_0,J^*J)\varphi_0,\psi_0\big\rangle_{\H_0}\\
&-\int_\Theta\d\theta\lim_{\varepsilon\searrow0}g_-(\varepsilon)
(1-\varepsilon)^{-1}\e^{-i\theta}\big\langle U_0V^*U
R\big((1-\varepsilon)^{-1}\e^{i\theta}\big)^*J
R_0\big((1-\varepsilon)^{-1}\e^{i\theta}\big)\varphi_0,\\
&\quad R_0\big((1-\varepsilon)^{-1}\e^{i\theta}\big)\psi_0\big\rangle_{\H_0}.
\end{align*}
Replacing this expression in \eqref{eq_first_term} and using the formulas
$$
R_0\big((1-\varepsilon)^\pm\e^{i\theta}\big)
=\pm g_+(\varepsilon)^{-1}\big(1-(1-\varepsilon)^{\pm1}\e^{-i\theta}U_0\big)
\delta_0(1-\varepsilon,\theta),
$$
we get after some steps the equality
\begin{align*}
&\big\langle E^U(\Theta)w_-(U,U_0,J)\varphi_0,
w_+(U,U_0,J)\psi_0\big\rangle_\H\nonumber\\
&=\big\langle E^{U_0}(\Theta)w_-(U_0,U_0,J^*J)\varphi_0,\psi_0\big\rangle_{\H_0}\\
&\quad-2\pi\int_\Theta\d\theta\lim_{\varepsilon\searrow0}\big\langle
V^*UR\big((1-\varepsilon)^{-1}\e^{i\theta}\big)^*J
\big(1-(1-\varepsilon)^{-1}\e^{-i\theta}U_0\big)\delta_0(1-\varepsilon,\theta)
\varphi_0,\delta_0(1-\varepsilon,\theta)\psi_0\big\rangle_{\H_0}.
\end{align*}
Now, a direct calculation using the formula
$$
R\big((1-\varepsilon)^{-1}\e^{i\theta}\big)^*
=-(1-\varepsilon)^{-1}\e^{i\theta}U^*R\big((1-\varepsilon)\e^{i\theta}\big)
$$
and the resolvent equation \eqref{eq_factors} shows that
\begin{align*}
&V^*UR\big((1-\varepsilon)^{-1}\e^{i\theta}\big)^*J
\big(1-(1-\varepsilon)^{-1}\e^{-i\theta}U_0\big)\\
&=V^*JU_0-V^*R\big((1-\varepsilon)^{-1}\e^{i\theta}\big)V\\
&=T_-\big((1-\varepsilon)\e^{i\theta}\big).
\end{align*}
Thus, we obtain
\begin{align*}
&\big\langle E^U(\Theta)w_-(U,U_0,J)\varphi_0,w_+(U,U_0,J)\psi_0\big\rangle_\H\\
&=\big\langle E^{U_0}(\Theta)w_-(U_0,U_0,J^*J)\varphi_0,\psi_0\big\rangle_{\H_0}\\
&\quad-2\pi\int_\Theta\d\theta\lim_{\varepsilon\searrow0}\big\langle
T_-\big((1-\varepsilon)\e^{i\theta}\big)\delta_0(1-\varepsilon,\theta)\varphi_0,
\delta_0(1-\varepsilon,\theta)\psi_0\big\rangle_{\H_0},
\end{align*}
which is the equality with the minus sign in \eqref{eq_pm}. Since the equality with
the plus sign is obtained in a similar way, this concludes the proof of the claim.

(b) Since the assumptions of Theorem \ref{thm_strong_wave} are satisfied, the
scattering operator $S(U,U_0,J)$ exists and satisfies
$S(U,U_0,J)=w_+(U,U_0,J)^*w_-(U,U_0,J)$. Furthermore, the operators $S(U,U_0,J)$ and
$w_\pm(U_0,U_0,J^*J)$ vanish on $\H_{\rm sc}(U_0)$ and have range in
$\H_{\rm ac}(U_0)$. Therefore, we obtain for any Borel set $\Theta\subset[0,2\pi)$ and
$\varphi_0,\psi_0\in\D_0$
\begin{align*}
&\big\langle E^U(\Theta)w_-(U,U_0,J)\varphi_0,w_+(U,U_0,J)\psi_0\big\rangle_\H
-\big\langle E^{U_0}(\Theta)w_\pm(U_0,U_0,J^*J)\varphi_0,\psi_0\big\rangle_{\H_0}\\
&=\big\langle F_0^{\rm(ac)}E^{U_0}(\Theta)S^{\rm(ac)}(U,U_0,J)
\big(F_0^{\rm(ac)}\big)^*F_0\varphi_0,F_0\psi_0\big\rangle_{\H_{\rm ac}(U)}\\
&\quad-\big\langle F_0^{\rm(ac)}E^{U_0}(\Theta)w_\pm^{\rm(ac)}(U_0,U_0,J^*J)
\big(F_0^{\rm(ac)}\big)^*F_0\varphi_0,F_0\psi_0\big\rangle_{\H_{\rm ac}(U_0)}\\
&=\int_{\widehat\sigma_0\cap\Theta}\d\theta\,
\big\langle\big(S(\theta)-u_\pm(\theta)\big)(F_0\varphi_0)(\theta),
(F_0\psi_0)(\theta)\big\rangle_{\h_0(\theta)}.
\end{align*}
Since $\Theta$ is arbitrary and $F_0=F_0P_{\rm ac}(U_0)$, we infer from this equation
and \eqref{eq_pm} that
\begin{align*}
&\big\langle\big(S(\theta)-u_\pm(\theta)\big)(F_0\varphi_0)(\theta),
(F_0\psi_0)(\theta)\big\rangle_{\h_0(\theta)}\\
&=\pm2\pi\lim_{\varepsilon\searrow0}\big\langle
T_\pm\big((1-\varepsilon)\e^{i\theta}\big)\delta_0(1-\varepsilon,\theta)
P_{\rm ac}(U_0)\varphi_0,\delta_0(1-\varepsilon,\theta)P_{\rm ac}(U_0)
\psi_0\big\rangle_{\H_0}\\
&=\pm2\pi\lim_{\varepsilon\searrow0}\big\langle P_{\rm ac}(U_0)
T_\pm\big((1-\varepsilon)\e^{i\theta}\big)P_{\rm ac}(U_0)\;\!
\delta_0(1-\varepsilon,\theta)\varphi_0,\delta_0(1-\varepsilon,\theta)
\psi_0\big\rangle_{\H_0},\quad\hbox{a.e. $\theta\in\widehat\sigma_0$,}
\end{align*}
as desired.
\end{proof}

The first term in the operator $T_+\big((1-\varepsilon)\e^{i\theta}\big)$ (resp.
$T_-\big((1-\varepsilon)\e^{i\theta}\big)$) can be factorised as $(GJU_0)^*G_0$ (resp.
$G_0^*(GJU_0)$). So, we need to determine conditions guaranteing that the operator
$GJU_0$ is weakly $U_0$-smooth\;\!:

\begin{Lemma}\label{lemma_GJU_0}
Assume that \eqref{eq_B_limit} is satisfied and that $G_0$ is weakly $U_0$-smooth.
Then, the operator $GJU_0$ is weakly $U_0$-smooth.
\end{Lemma}

\begin{proof}
Using \eqref{eq_GJ} with $z=(1-\varepsilon)^{\pm1}\e^{i\theta}$, we obtain
\begin{align*}
&\big|g_\pm(\varepsilon)\big|^{1/2}\;\!
\big\|GJR_0\big((1-\varepsilon)^{\pm1}\e^{i\theta}\big)\big\|_{\B(\H_0,\G)}\\
&\le\left(\big|g_\pm(\varepsilon)\big|^{1/2}\;\!
\big\|GR\big((1-\varepsilon)^{\pm1}\e^{i\theta}\big)\big\|_{\B(\H,\G)}\right)
\|J\|_{\B(\H_0,\H)}\\
&\quad+\big\|B\big((1-\varepsilon)^{\mp1}\e^{i\theta}\big)\big\|_{\B(\G)}
\left(\big|g_\pm(\varepsilon)\big|^{1/2}\;\!\big\|G_0
R_0\big((1-\varepsilon)^{\pm1}\e^{i\theta}\big)\big\|_{\B(\H_0,\G)}\right),
\end{align*}
and the assumption \eqref{eq_B_limit} implies that $G$ is weakly $U$-smooth.
Therefore, all the factors on the r.h.s. above are uniformly bounded in
$\varepsilon\in(0,1)$ as a consequence of \eqref{eq_B_limit} and an application of
Lemma \ref{lemma_weak}(a) for the pairs $(U_0,G_0)$ and $(U,G)$. In consequence, the
bound \eqref{eq_estimate_2} is satisfied with $T_0=GJ$, and Lemma \ref{lemma_weak}(a)
implies that the operator $GJ$ is weakly $U_0$-smooth. Therefore, the operator $GJU_0$
is weakly $U_0$-smooth too.
\end{proof}

Using what precedes, we can conclude in the next theorem the derivation of the
representation formulas for the scattering matrix. We recall that the operators
$Z_0(\theta,\;\!\cdot\;\!)$ and $B_\pm(\theta)$ have been defined in Lemma
\ref{lemma_F_0}(b) and Equation \eqref{eq_B_limit}, respectively.

\begin{Theorem}[Scattering matrix]\label{thm_S_matrix}
Assume that \eqref{eq_free_lim} is satisfied for a dense set $\D_0\subset\H_0$, that
\eqref{eq_B_limit} is verified, and that $G_0$ is weakly $U_0$-smooth. Then, we have
for a.e. $\theta\in\widehat\sigma_0$ the representation formulas for the scattering
matrix\;\!:
\begin{align}
S(\theta)&=u_+(\theta)+2\pi\big(Z_0(\theta,GJU_0)Z_0(\theta,G_0)^*
-Z_0(\theta,G_0)B_+(\theta)Z_0(\theta,G_0)^*\big),\label{eq_S_matrix_1}\\
S(\theta)&=u_-(\theta)-2\pi\big(Z_0(\theta,G_0)Z_0(\theta,GJU_0)^*
-Z_0(\theta,G_0)B_-(\theta)Z_0(\theta,G_0)^*\big).\label{eq_S_matrix_2}
\end{align}
\end{Theorem}

\begin{proof}
We only give the proof of \eqref{eq_S_matrix_1}, since the proof of
\eqref{eq_S_matrix_2} is similar.

Equation \eqref{eq_pm_bis} holds under our assumptions. Therefore, thanks to Lemma
\ref{lemma_F_0}(a), in order to prove \eqref{eq_S_matrix_1} it is sufficient to show
for any $\varphi_0,\psi_0\in\D_0$ and a.e. $\theta\in\widehat\sigma_0$ that
\begin{align}
&\lim_{\varepsilon\searrow0}\big\langle
P_{\rm ac}(U_0)T_+\big((1-\varepsilon)\e^{i\theta}\big)P_{\rm ac}(U_0)\;\!
\delta_0(1-\varepsilon,\theta)\varphi_0,\delta_0(1-\varepsilon,\theta)\psi_0
\big\rangle_{\H_0}\nonumber\\
&=\big\langle\big(Z_0(\theta,GJU_0)Z_0(\theta,G_0)^*
-Z_0(\theta,G_0)B_+(\theta)Z_0(\theta,G_0)^*\big)(F_0\varphi_0)(\theta),
(F_0\psi_0)(\theta)\big\rangle_{\H_0}.\label{eq_to_prove}
\end{align}
The first term in the operator $T_+\big((1-\varepsilon)\e^{i\theta}\big)$ gives the
contribution
\begin{equation}\label{eq_cont_1}
\lim_{\varepsilon\searrow0}\big\langle P_{\rm ac}(U_0)(GJU_0)^*G_0P_{\rm ac}(U_0)\;\!
\delta_0(1-\varepsilon,\theta)\varphi_0,\delta_0(1-\varepsilon,\theta)\psi_0
\big\rangle_{\H_0}.
\end{equation}
Since $GJU_0$ is weakly $U_0$-smooth by Lemma \ref{lemma_GJU_0}, $G_0$ is weakly
$U_0$-smooth by assumption, and the limit
$\slim_{\varepsilon\searrow0}G_0\delta_0(1-\varepsilon,\theta)\varphi_0$ exists a.e.
$\theta\in[0,2\pi)$ due to Remark \ref{rem_ass_equiv}(a), we can apply Lemma
\ref{lemma_F_0}(c), equation \eqref{eq_lim_kernel}, to infer that
\begin{equation}\label{eq_result_1}
\eqref{eq_cont_1}
=\big\langle Z_0(\theta,GJU_0)Z_0(\theta,G_0)^*(F_0\varphi_0)(\theta),
(F_0\psi_0)(\theta)\big\rangle_{\h_0(\theta)},
\quad\hbox{a.e. $\theta\in\widehat\sigma_0$.}
\end{equation}
Similarly, the second term in the operator $T_+\big((1-\varepsilon)\e^{i\theta}\big)$
gives the contribution
\begin{equation}\label{eq_cont_2}
\lim_{\varepsilon\searrow0}\big\langle
P_{\rm ac}(U_0)G_0^*B\big((1-\varepsilon)\e^{i\theta}\big)G_0P_{\rm ac}(U_0)\;\!
\delta_0(1-\varepsilon,\theta)\varphi_0,\delta_0(1-\varepsilon,\theta)\psi_0
\big\rangle_{\H_0}.
\end{equation}
Since $G_0$ is weakly $U_0$-smooth by assumption, the limits
$\slim_{\varepsilon\searrow0}G_0\delta_0(1-\varepsilon,\theta)\varphi_0$ and
$\slim_{\varepsilon\searrow0}G_0\delta_0(1-\varepsilon,\theta)\psi_0$ exist a.e.
$\theta\in[0,2\pi)$ due to Remark \ref{rem_ass_equiv}(a), and
$\wlim_{\varepsilon\searrow0}B\big((1-\varepsilon)\e^{i\theta}\big)=B_+(\theta)$ for
a.e. $\theta\in[0,2\pi)$ due to \eqref{eq_B_limit}, we can apply Lemma
\ref{lemma_F_0}(c), equation \eqref{eq_lim_kernel_bis}, to infer that
\begin{equation}\label{eq_result_2}
\eqref{eq_cont_2}
=\big\langle Z_0(\theta,G_0)B_+(\theta)Z_0(\theta,G_0)^*(F_0\varphi_0)(\theta),
(F_0\psi_0)(\theta)\big\rangle_{\h_0(\theta)},
\quad\hbox{a.e. $\theta\in\widehat\sigma_0$.}
\end{equation}
Combining \eqref{eq_result_1} and \eqref{eq_result_2}, we obtain \eqref{eq_to_prove}
as desired.
\end{proof}

As in the self-adjoint case, it is possible to establish various variants of the
representation formulas \eqref{eq_S_matrix_1}-\eqref{eq_S_matrix_2}. But we prefer not
to do it here for the sake of conciseness. Instead, as a closing observation for the
section, we recall that in the one-Hilbert space case $\H_0=\H$ and $J=1_{\H_0}$ we
have $u_\pm(\theta)=1_{\h_0(\theta)}$ for a.e. $\theta\in\widehat\sigma_0$. Therefore,
in such a case, the formulas \eqref{eq_S_matrix_1}-\eqref{eq_S_matrix_2} reduce to
\begin{align*}
S(\theta)&=1_{\h_0(\theta)}+2\pi\big(Z_0(\theta,GU_0)Z_0(\theta,G_0)^*
-Z_0(\theta,G_0)B_+(\theta)Z_0(\theta,G_0)^*\big),\\
S(\theta)&=1_{\h_0(\theta)}-2\pi\big(Z_0(\theta,G_0)Z_0(\theta,GU_0)^*
-Z_0(\theta,G_0)B_-(\theta)Z_0(\theta,G_0)^*\big),
\end{align*}
for a.e. $\theta\in\widehat\sigma_0$

%--------------------------------------------------------------------------------------
\section{Application to anisotropic quantum walks}\label{section_walks}
\setcounter{equation}{0}
%--------------------------------------------------------------------------------------

In this section, we illustrate the theory of Sections
\ref{section_wave}-\ref{section_matrix} by deriving representation formulas for the
wave operators and the scattering matrix of quantum walks with an anisotropic coin. We
start by recalling the definition of quantum walks with an anisotropic coin, as
presented in \cite{RST_2018,RST_2019}.

Let $\H$ be the Hilbert space of square-summable $\C^2$-valued sequences
$$
\H:=\ell^2(\Z,\C^2)
=\left\{\Psi:\Z\to\C^2\mid\sum_{x\in\Z}\|\Psi(x)\|_2^2<\infty\right\},
$$
with $\|\cdot\|_2$ the usual norm on $\C^2$. Then, the evolution operator of the
one-dimensional quantum walk that we consider is defined as $U:=SC$, with $S$ the
shift operator given by
$$
(S\Psi)(x)
:=\begin{pmatrix}
\Psi^{(0)}(x+1)\\
\Psi^{(1)}(x-1)
\end{pmatrix},
\quad
\Psi
=\begin{pmatrix}
\Psi^{(0)}\\
\Psi^{(1)}
\end{pmatrix}\in\H,~x\in\Z,
$$
and $C$ the coin operator given by
$$
(C\Psi)(x):=C(x)\Psi(x),\quad\Psi\in\H,~x\in\Z,~C(x)\in\U(2).
$$
The coin operator $C$ is assumed to have an anisotropic behaviour at infinity; it
converges with short-range rate to two asymptotic coin operators, one on the left and
one on the right, in the following way\;\!:

\begin{Assumption}[Short-range]\label{ass_short}
There exist $C_\ell,C_{\rm r}\in\U(2)$, $\kappa_\ell,\kappa_{\rm r}>0$, and
$\varepsilon_\ell,\varepsilon_{\rm r}>0$ such that
\begin{align*}
&\big\|C(x)-C_\ell\big\|_{\B(\C^2)}
\le\kappa_\ell\;\!|x|^{-1-\varepsilon_\ell}\quad\mathrm{if}~x<0\\
&\big\|C(x)-C_{\rm r}\big\|_{\B(\C^2)}
\le\kappa_{\rm r}\;\!|x|^{-1-\varepsilon_{\rm r}}\quad\mathrm{if}~x>0,
\end{align*}
where the indexes $\ell$ and ${\rm r}$ stand for ``left" and ``right".
\end{Assumption}

This assumption provides two new unitary operators $U_\star:=SC_\star$
($\star=\ell,{\rm r}$), with resolvent $R_\star$, describing the asymptotic behaviour
of $U$ on the left and on the right. It also motivates to define the free evolution
operator as the direct sum operator $U_0:=U_\ell\oplus U_{\rm r}$ in the Hilbert space
$\H_0:=\H\oplus\H$. The identification operator $J\in\B(\H_0,\H)$ is given by
$$
J(\Psi_0):=j_\ell\;\!\Psi_{0,\ell}+j_{\rm r}\;\!\Psi_{0,\rm r},
\quad\Psi_0=(\Psi_{0,\ell},\Psi_{0,\rm r})\in\H_0,
$$
with
$$
j_{\rm r}(x):=
\begin{cases}
1 & \hbox{if $x\ge0$}\\
0 & \hbox{if $x\le-1$}
\end{cases}
\quad\hbox{and}\quad
j_\ell:=1-j_{\rm r}.
$$	

The nature of the spectrum of $U_\star$ and $U_0$ depends on the choice of the
matrices $C_\star\in\U(2)$. For simplicity, we consider here only matrices $C_\star$
such that $U_\star$ and $U_0$ have purely absolutely continuous spectrum. Namely,
first we parameterise the matrices $C_\star$ as
$$
C_\star=\e^{i\delta_\star/2}
\begin{pmatrix}
a_\star\e^{i(\alpha_\star-\delta_\star/2)}
& b_\star\e^{i(\beta_\star-\delta_\star/2)}\\
-b_\star\e^{-i(\beta_\star-\delta_\star/2)}
& a_\star\e^{-i(\alpha_\star-\delta_\star/2)}
\end{pmatrix}
$$
with $a_\star,b_\star\in[0,1]$ satisfying $a_\star^2+b_\star^2=1$, and
$\alpha_\star,\beta_\star,\delta_\star\in(-\pi,\pi]$. Then, we note from Theorem 2.2,
Proposition 4.5(c) and Lemma 4.6(d) of \cite{RST_2018} that the operators $U_\star$
and $U_0$ have purely absolutely continuous spectra if $a_\star\in(0,1]$, i.e.,
$$
\sigma(U_\star)=\sigma_{\rm ac}(U_\star)
\quad\hbox{and}\quad
\sigma(U_0)
=\sigma_{\rm ac}(U_0)
=\sigma_{\rm ac}(U_\ell)\cup\sigma_{\rm ac}(U_{\rm r}),
$$
and that the essential spectrum of $U$ coincides with the spectrum of $U_0$\;\!:
$$
\sigma_{\rm ess}(U)=\sigma(U_\ell)\cup\sigma(U_{\rm r})=\sigma(U_0).
$$
Furthermore, if we set
$
\tau(U):=\partial\sigma(U_\ell)\cup\partial\sigma(U_{\rm r}),
$
with $\partial\sigma(U_\star)$ the boundaries of $\sigma(U_\star)$ in $[0,2\pi)$, then
for any closed set $\Theta\subset[0,2\pi)\setminus\tau(U)$, the operator $U$ has at
most finitely many eigenvalues in $\Theta$, each one of finite multiplicity, and $U$
has no singular continuous spectrum in $\Theta$ \cite[Thm.~2.4]{RST_2018}.

In \cite[Thm.~3.3 \& Prop.~3.4]{RST_2019}, it has also been shown that the strong wave
operators
$$
W_\pm(U,U_0,J)
:=\slim_{n\to\pm\infty}U^nJU_0^{-n}P_{\rm ac}(U_0)
=\slim_{n\to\pm\infty}U^nJU_0^{-n}
$$
exist and are complete.\footnote{In \cite[Thm.~3.3]{RST_2019} it has been shown that
$
\overline{\Ran\big(W_\pm(U,U_0,J)\big)}=\H_{\rm ac}(U)
$,
and in \cite[Prop.~3.4]{RST_2019} it has been shown that the operators
$W_\pm(U,U_0,J)$ are partial isometries with initial sets $\H_0^\pm\subset\H_0$. Since
partial isometries have closed ranges, it follows that the wave operators are
complete, i.e. $\Ran\big(W_\pm(U,U_0,J)\big)=\H_{\rm ac}(U)$. As a by-product, we
infer that the scattering operator $S(U,U_0,J)=W_+(U,U_0,J)^*\;\!W_-(U,U_0,J)$ is
unitary from $\H_0^-$ to $\H_0^+$.} In addition, the relations
$$
W_\pm(U,U_0,J)^*=W_\pm(U_0,U,J^*)
\quad\hbox{and}\quad
W_\pm(U,U_0,J)\;\!\eta(U_0)=\eta(U)\;\!W_\pm(U,U_0,J)
$$
hold for each bounded Borel function $\eta:\S^1\to\C$. To go further and show that
the strong wave operators coincide with the stationary wave operators, we first need
a lemma on the factorisation of the perturbation $V=JU_0-UJ$. We use the notation
$Q$ for the position operator in $\H$\;\!:
$$
(Q\Psi)(x):=x\;\!\Psi(x),
\quad x\in\Z,~\Psi\in\dom(Q):=\big\{\Psi\in\H\mid\|Q\Psi\|_\H<\infty\big\}.
$$

\begin{Lemma}\label{lemma_V_trace}
Suppose that Assumption \ref{ass_short} holds. Then, the perturbation $V$ factorises
as $V=G^*G_0$, with $G_0\in S_2(\H_0)$ and $G\in S_2(\H,\H_0)$.
\end{Lemma}

\begin{proof}
We know from the proof of \cite[Thm.~3.3]{RST_2019} that there exist $s>1/2$ and
$D\in\B(\H_0,\H)$ such that the operators
\begin{equation}\label{def_G_0_G}
G_0:=\langle Q\rangle^{-s}\oplus\langle Q\rangle^{-s}
\quad\hbox{and}\quad
G:=D^*\langle Q\rangle^{-s}
\end{equation}
satisfy $V=G^*G_0$. Therefore, in order to prove the claim, it is sufficient to show
that $\langle Q\rangle^{-s}\in S_2(\H)$. Let $(v_y)_{y\in\Z}$ be the orthonormal
basis of $\ell^2(\Z)$ given by $v_y(x):=\delta_{xy}$ ($x\in\Z$) and let
$e_1:=\big(\begin{smallmatrix}1\\0\end{smallmatrix}\big)$,
$e_2:=\big(\begin{smallmatrix}0\\1\end{smallmatrix}\big)$ be the standard basis of
$\C^2$. Then, the family $(v_y\otimes e_i)_{y\in\Z,\,i=1,2}$ is an orthonormal basis
of $\H$, and a direct calculation gives
$$
\big\|\langle Q\rangle^{-s}\big\|_{S_2(\H)}^2
=2\sum_{x\in\Z}\langle x\rangle^{-2s}
<\infty
$$
which proves the claim.
\end{proof}

In the next theorem, we show that the strong wave operators coincide with the
stationary wave operators, and we give representation formulas for the stationary wave
operators and the scattering matrix.

\begin{Theorem}\label{thm_walks}
Suppose that Assumption \ref{ass_short} holds and that $a_\star\in(0,1]$.
\begin{enumerate}
\item[(a)] The strong wave operators $W_\pm(U,U_0,J)$ coincide with the stationary
wave operators $w_\pm(U,U_0,J)$, and we have for any
$\Psi_0=(\Psi_{0,\ell},\Psi_{0,\rm r})\in\H_0$ and $\psi\in\H$ the representation
formulas
\begin{align}
&\big\langle w_\pm(U,U_0,J)\psi_0,\psi\big\rangle_\H\label{eq_stat_walk}\\
&=\pm\int_0^{2\pi}\d\theta\,\lim_{\varepsilon\searrow0}g_\pm(\varepsilon)\;\!
\big\langle j_\ell R_\ell\big((1-\varepsilon)^{\pm1}\e^{i\theta}\big)\Psi_{0,\ell}
+j_{\rm r}R_{\rm r}\big((1-\varepsilon)^{\pm1}\e^{i\theta}\big)\Psi_{0,\rm r},
R\big((1-\varepsilon)^{\pm1}\e^{i\theta}\big)\psi\big\rangle_\H.\nonumber
\end{align}

\item[(b)] We have for a.e. $\theta\in\widehat\sigma_0$ the representation formulas
for the scattering matrix
\begin{align*}
S(\theta)
&=u_+(\theta)+2\pi\big(Z_0(\theta,GJU_0)Z_0(\theta,G_0)^*
-Z_0(\theta,G_0)B_+(\theta)Z_0(\theta,G_0)^*\big)\\
S(\theta)
&=u_-(\theta)-2\pi\big(Z_0(\theta,G_0)Z_0(\theta,GJU_0)^*
-Z_0(\theta,G_0)B_-(\theta)Z_0(\theta,G_0)^*\big).
\end{align*}
with $G_0,G$ as in \eqref{def_G_0_G}.
\end{enumerate}
\end{Theorem}

\begin{proof}
(a) We know from Lemma \ref{lemma_V_trace} that $V=G^*G_0$ with $G_0\in S_2(\H_0)$ and
$G\in S_2(\H,\H_0)$. So, it follows from Example \ref{ex_trace_class} that the assumptions
of Theorem \ref{thm_strong_wave} are satisfied for the set $\D_0=\H_0$. Since Theorem
\ref{thm_stat_wave} also holds under the assumptions of Theorem \ref{thm_strong_wave},
we obtain that the strong wave operators $W_\pm(U,U_0,J)$ coincide with the stationary
wave operators $w_\pm(U,U_0,J)$, and that
$$
\big\langle w_\pm(U,U_0,J)\psi_0,\psi\big\rangle_\H
=\pm\int_0^{2\pi}\d\theta\,\lim_{\varepsilon\searrow0}g_\pm(\varepsilon)\;\!
\big\langle JR_0\big((1-\varepsilon)^{\pm1}\e^{i\theta}\big)\psi_0,
R\big((1-\varepsilon)^{\pm1}\e^{i\theta}\big)\psi\big\rangle_\H.
$$
Now, this equation implies \eqref{eq_stat_walk} if one takes into account the
definitions of $J$, $R_0=R_\ell\oplus R_{\rm r}$ and $\Psi_0$.

(b) We know from point (a) that the assumptions of Theorem \ref{thm_strong_wave} are
satisfied for the set $\D_0=\H_0$. We also know from Lemma \ref{lemma_V_trace} that
$G_0\in S_2(\H_0)$. Thus, Lemma \ref{lemma_HS}(b) implies that $G_0$ is weakly
$U_0$-smooth. Therefore, all the assumptions of Theorem \ref{thm_S_matrix} are
satisfied, and the claim is a direct consequence of that theorem.
\end{proof}

%--------------------------------------------------------------------------------------
%\bibliography{../bibliographie/bibliographie}
%--------------------------------------------------------------------------------------

\def\cprime{$'$} \def\polhk#1{\setbox0=\hbox{#1}{\ooalign{\hidewidth
  \lower1.5ex\hbox{`}\hidewidth\crcr\unhbox0}}}
  \def\polhk#1{\setbox0=\hbox{#1}{\ooalign{\hidewidth
  \lower1.5ex\hbox{`}\hidewidth\crcr\unhbox0}}}
  \def\polhk#1{\setbox0=\hbox{#1}{\ooalign{\hidewidth
  \lower1.5ex\hbox{`}\hidewidth\crcr\unhbox0}}} \def\cprime{$'$}
  \def\cprime{$'$} \def\polhk#1{\setbox0=\hbox{#1}{\ooalign{\hidewidth
  \lower1.5ex\hbox{`}\hidewidth\crcr\unhbox0}}}
  \def\polhk#1{\setbox0=\hbox{#1}{\ooalign{\hidewidth
  \lower1.5ex\hbox{`}\hidewidth\crcr\unhbox0}}} \def\cprime{$'$}
  \def\cprime{$'$} \def\cprime{$'$}

%--------------------------------------------------------------------------------------

\end{document}